\documentclass[letter,11pt]{article}

\usepackage{paralist}
\usepackage{booktabs}
\usepackage{color}

\DeclareMathAlphabet{\mathbfsf}{\encodingdefault}{\sfdefault}{bx}{n}

\usepackage{amsmath,amsthm,algorithmicx,amssymb,subfigure,graphicx,url,multirow}
\usepackage{algorithm}
\usepackage[noend]{algpseudocode}
\usepackage{enumerate}
\usepackage{enumitem}
\usepackage{mathtools}

\usepackage{thmtools}
\usepackage{thm-restate}
\usepackage{xcolor}

\newtheorem{invariant}{Invariant}
\newcommand{\batchparallelcap}[0]{Batch-Parallel}
\newcommand{\batchparallel}[0]{batch-parallel}
\newcommand{\batchdynamic}[0]{batch-dynamic}
\newcommand{\batchdynamiccap}[0]{Batch-Dynamic}

\newcommand{\Boruvka}{Bor\r{u}vka}

\newcommand{\parlevelsearch}{\textsc{ParallelLevelSearch}}
\newcommand{\intlevelsearch}{\textsc{InterleavedLevel\-Search}}

\newcommand{\defn}[1]{\emph{\textbf{#1}}}
\newcommand{\codevar}[1]{\mathit{#1}}

\newcommand{\myparagraph}[1]{\smallskip\noindent {\bf #1.}}

\newcommand{\id}[1]{\ifmmode\mathit{#1}\else\textit{#1}\fi}
\newcommand{\const}[1]{\ifmmode\mbox{\textc{#1}}\else\textsc{#1}\fi}

\newcommand{\assign}{\gets}

\newcommand{\mpram}{$\mathsf{TRAM}$}
\newcommand{\pram}{$\mathsf{PRAM}$}
\newcommand{\crcwpram}{$\mathsf{CRCW}\ \mathsf{PRAM}$}
\newcommand{\process}{thread}
\newcommand{\processes}{threads}

\newcommand{\forkins}{\texttt{fork}}
\newcommand{\insend}{\texttt{end}}

\newcommand{\bark}{\Delta}

\renewcommand{\log}[1]{\mathop{\textnormal{lg}} #1}

\algrenewcommand\algorithmicindent{1em}

\algnewcommand{\algcomment}[1]{\hfill{\color{purple!40!black}\emph{// #1}}}    % Trailing comment
\algnewcommand{\alglinecomment}[1]{{\color{purple!40!black}\emph{// #1}}}      % Full line comment

\usepackage{balance}
\usepackage{microtype}

\mathtoolsset{showonlyrefs}

\usepackage{geometry}
\newtheorem{theorem}{Theorem}[section]
\newtheorem{lemma}{Lemma}[section]

\begin{document}

\title{Parallel \batchdynamiccap{} Graph Connectivity\footnote{This is the full version of the paper appearing in the ACM Symposium on Parallelism in Algorithms and Architectures (SPAA), 2019}}

\author{\normalsize Umut A. Acar \\ \small Carnegie Mellon University \\ \small umut@cs.cmu.edu \and \normalsize Daniel Anderson \\ \small Carnegie Mellon University \\ \small dlanders@cs.cmu.edu \and \normalsize Guy E. Blelloch \\ \small Carnegie Mellon University \\ \small guyb@cs.cmu.edu \and \normalsize Laxman Dhulipala \\ \small Carnegie Mellon University \\ \small ldhulipa@cs.cmu.edu}
\date{}

\maketitle

\begin{abstract}
  % With the rapid growth of graph datasets over the past decade, a new
  % kind of dynamic algorithm, supporting the ability to ingest
  % \emph{batches of updates} and \emph{exploit parallelism} is needed
  % in order to efficiently process large streams of updates. 
  In this
  paper, we study batch parallel algorithms for the dynamic
  connectivity problem, a fundamental problem that has received
  considerable attention in the sequential setting.    The most well known 
  sequential algorithm for dynamic connectivity is the elegant level-set algorithm of Holm, de
  Lichtenberg and Thorup (HDT), which achieves $O(\log^2 n)$ amortized
  time per edge insertion or deletion, and $O(\log n / \log\log n)$ time per query.

%  \umut{faster?}
%\umut{$\bar{x}$ is hard to read, how about $\kappa$}

  We design a parallel batch-dynamic connectivity
  algorithm that is work-efficient with respect to the HDT algorithm
  for small batch sizes, and is asymptotically faster when the average
  batch size is sufficiently large.  Given a sequence of batched
  updates, where $\bark{}$
  is the average batch size of all deletions, our algorithm achieves
  $O(\log n \log(1 + n / \bark{}))$ expected amortized work per edge
  insertion and deletion and $O(\log^3 n)$ depth w.h.p. Our algorithm
  answers a batch of $k$ connectivity queries in $O(k \log(1 + n/k))$
  expected work and $O(\log n)$ depth w.h.p.  To the best of our knowledge, our
  algorithm is the first parallel batch-dynamic algorithm for
  connectivity.
\end{abstract}

\clearpage

\section{Introduction}\label{sec:intro}

%\umut{fundamental problem}
Computing the connected components of a graph is a fundamental
problem that has been studied in many different models of
computation~\cite{tarjan1975efficiency, ShiloachV82, reingold08connectivity,
  holm2001poly, ahn2012connectivity, andoniparallel}.
%\umut{$n$ and $m$ are not used, drop.}
The \defn{connectivity problem} takes as input an undirected graph $G$
and requires an assignment of labels to vertices such that two vertices have the same
label if and only if they are in the same connected component.
%\umut{is to maintain \ra requires  maintaining.}
%\umut{$n$ vertex unused, drop.} \laxman{I think this is important to
%keep since the number of vertices is defined to be exactly n}
The dynamic version of the problem requires maintaining a data
structure over an $n$ vertex undirected graph that supports 
insertions and deletions of edges, and queries
of whether two vertices are in the same connected component.
%\umut{Despite the progress makes it sound like we are going to make
%more progress.  I would say, there has been all this work, all
%consider single updates.} \laxman{does this sound better?}
Despite the large body of work on the dynamic connectivity problem over the past
two decades~\cite{henzinger1995randomized, eppstein1997sparsification,
  holm2001poly, thorup1999decremental, thorup2000near,
  henzinger2001maintaining, wulff2013faster, kapron2013dynamic,
  rasumssen16faster, huang17dynconn, nanongkai2017dynamic, wulff2017fully}, little is
known about \batchdynamic{} connectivity algorithms that process
\emph{batches of queries and updates}, either sequentially or in parallel.

Traditional dynamic algorithms were motivated by applications
where data undergos small changes that
can be adequately handled by updates of single elements. Today, however,
applications operate on increasingly large datasets that undergo
rapid changes over time: for example, millions of individuals can
simultaneously interact with a web site, make phone calls, send emails
and so on. In the context of these applications, traditional dynamic
algorithms require serializing the changes made and processing them
one at a time, missing an opportunity to exploit the parallelism
afforded by processing batches of changes.

%% \umut{Need transition. Something like this? Blend...
%%
%% The singly dynamic algorithms have been motivated by applications from the early days of computing where data would change by small amounts at a time, which can be adequately handled by these algorithms.
%% %
%% Today's applications operate in more dynamic environments, where data changes more rapidly: for example, thousands and millions of customers can log on to a web site at the same time, make phone calls at the same time, send an email at the same time and so on.
%% %
%% In such applications, the singly dynamic algorithms require serializing the changes and processing them one by time, missing thus the opportunity to exploit the usually natural parallelism afforded by the batch of changes.
%% %
%% }

Motivated by such applications, there has been recent
interest in developing theoretically efficient parallel \batchdynamic{}
algorithms~\cite{acar2011parallelism, simsiri2016work,  acar2017brief,
  tseng2018batch}.
In the \batchdynamic{} setting, instead of applying one update or
query at a time, a whole batch is applied.  A batch could be of size
$\lg{n}$, $\sqrt{n}$, or $n/ \lg n$ for example.  There are two
advantages of applying operations in batches.
\vspace{0.5em}
\begin{enumerate}[ref={\arabic*-}, topsep=0pt,itemsep=0ex,partopsep=0ex,parsep=1ex, leftmargin=*]
\item Batching operations allows for more parallelism.
\item Batching operations can reduce the cost of each update.
\end{enumerate}
\vspace{0.5em}
 In this paper we are interested in both these advantages.
We use the term parallel \batchdynamic{} to mean algorithms that process
batches of operations instead of single ones, and for which the
algorithm itself is parallel. The underlying parallel model used in
this paper is a formalization of the widely used shared-memory work-depth
model~\cite{Blelloch96,blumofe1999scheduling,BGM99,ABP01}.
% Guy : I think the following is an interesting point, but get in the way
% Note that parallel \batchdynamic{} algorithms can be of
% interest even for batches of size 1, since a parallel algorithm with
% low worst-case depth will result in lower running times for an
% amortized dynamic algorithm that may perform a large amount of work
% for some updates or queries.
%Both of these results are also parallel \batchdynamic{} algorithms.

Understanding the connectivity structure of graphs is of significant
practical interest, for example, due to its use as a primitive for
clustering the vertices of a graph~\cite{vldbsurvey}.  Due to the
importance of connectivity there are several implementations of
parallel \batchdynamic{} connectivity algorithms~\cite{mccoll2013new,
  iyer2015cellq, Wickramaarachchi15, Iyer2016, Sha2017, Vora2017}. In
the worst case, however, these algorithms may recompute the connected
components of the entire graph even for very small batches. Since this
requires $O(m + n)$ work, it makes the worst-case performance of the
algorithms no better than running a static parallel algorithm.  On the
theoretical side, existing \batchdynamic{} efficient connectivity
algorithms have only been designed for restricted settings, e.g., in
the incremental setting when all updates are edge
insertions~\cite{simsiri2016work}, or when the underlying graph is a
forest~\cite{reif94treecontraction, acar2017brief, tseng2018batch}.
Therefore, two important questions are:

\begin{enumerate}[ref={\arabic*-}, leftmargin=*]
  \item \emph{Is there a \batchdynamic{} connectivity algorithm that is
  asymptotically faster than existing dynamic connectivity algorithms
  for large enough batches of insertions, deletions and queries?}
  \item \emph{Can the \batchdynamic{} connectivity algorithm be
      parallelized to achieve low worst-case depth?}
\end{enumerate}

\noindent In this paper we give an algorithm that answers both of these
questions affirmatively. To simplify our exposition and present the
main ideas, we first give a less efficient version of the algorithm
that runs in $O(\log^4 n)$ depth w.h.p.\footnote{ We say that an
  algorithm has $O(f(n))$ cost \defn{with high probability (w.h.p.)}
  if it has $O(k\cdot f(n))$ cost with probability at least
  $1 - 1/n^{k}$, $k \geq 1$.  } and performs $O(\log^2 n)$ expected
amortized work per update, making it work-efficient with respect to
the sequential algorithm of Holm, de Lichtenberg, and Thorup.  Next,
we describe the improved algorithm which runs in $O(\log^3 n)$ depth
w.h.p.
and achieves an improved work bound that is asymptotically
faster than the HDT algorithm for sufficiently large batch sizes. We
note that our depth bounds hold even when processing the updates one a
time, ignoring batching. Our improved work bounds are derived by a
novel analysis of the work performed by the algorithm over all batches
of deletions.
Our contribution is summarized by the following theorem:
\begin{theorem}
There is a parallel \batchdynamic{} data structure which, given
batches of edge insertions, deletions, and connectivity queries
processes all updates in $O\left(\log n \log\left( 1 +
    \frac{n}{\bark{}} \right) \right)$ expected amortized work per
edge insertion or deletion where $\bark{}$ is the average batch size
of a deletion operation. The cost of connectivity queries is $O(k
\log(1 + n/k))$ work and $O(\log n)$ depth for a batch of $k$ queries.
The depth to process a batch of edge insertions and deletions is
$O(\log n)$ and $O(\log^3 n)$ respectively.
\end{theorem}\label{thm:parallel-bounds}

\subsection{Technical Overview}
The starting point of our algorithm is the classic Holm, de
Lichtenberg and Thorup (HDT) dynamic connectivity
algorithm~\cite{holm2001poly}\footnote{We provide full details of the HDT
algorithm in Section~\ref{sec:hdt}.}.  Like nearly all existing dynamic
connectivity algorithms, the HDT algorithm maintains a spanning forest
certifying the connectivity of the graph. The algorithm maintains a
set of $\log n$ nested forests under two carefully designed
invariants. The forests are represented using the Euler tour tree
(ET-tree) data
structure~\cite{miltersen1994complexity,henzinger1995randomized}.
%% due to \guy{a little odd to give names
%%   instead of just citations for subroutines.} Henzinger and
%% King~\cite{henzinger1995randomized}, and Miltersen et
%% al.~\cite{miltersen1994complexity}.

The main challenge in a dynamic connectivity algorithm is to
efficiently find a \emph{replacement edge}, or a non-tree edge going
between the two disconnected components after deleting a tree edge.
The key idea of the HDT algorithm is to organize the spanning-forest
of the graph into $\log n$ levels of trees. The top-most level of the
structure stores a spanning forest of the entire graph, and each level
contains all tree-edges stored in levels below it. The algorithm
ensures that the largest size of a component at level $i$ is $2^{i}$.
Using these invariants, the algorithm is able to cleverly search the
tree edges so that each non-tree edge is examined at most $\log n$
times as a candidate replacement edge.
%% The
%% algorithm searches the tree edges so that that each non-tree edge is
%% considered at most $\log n$ times as a candidate replacement edge.
%% \guy{need to say something about finer granularity at lower levels,
%%   otherwise makes little sense.}
The main idea is to store each non-tree edge at a single level
(initially the top-most level), and push the edge to a lower level
each time it is unsuccessfully considered as a replacement edge. Since
there are $\log n$ levels, and the cost of discovering, processing,
and removing an edge from each level using ET-tree operations is
$O(\log n)$, the amortized cost of the HDT algorithm is $O(\log^2 n)$
per edge operation. We now discuss the main challenges and sequential
bottlenecks that arise in the HDT algorithm that a parallel
\batchdynamic{} algorithm must address.

\myparagraph{Efficiently searching for replacements} A challenge, and
sequential bottleneck in the HDT algorithm is the fact that it
processes each non-tree edge one at a time---a property which is
crucial for achieving good amortized bounds.
%\guy{claimed?
%  Makes it sound like it is not true, sort of like Trump making claims.}
Aside from hindering parallelism, processing the edges one at a time
eliminates any potential for improved batch bounds, since finding the
representative of the endpoints of an edge costs $O(\log n)$ time per
query. Therefore,
%\guy{``in order'' has zero information content--almost always drop it.}
to obtain an efficient batch or parallel
algorithm we must examine batches of \emph{multiple} non-tree edges at
a time, while also ensuring that we do not perform extra work that
cannot be charged to level-decreases on an edge. Our approach is to
use a \emph{doubling technique}, where we examine sets of non-tree
edges with geometrically increasing sizes.
%% The approach lets us obtain
%% polylogarithmic depth algorithms that can still charge all work they
%% perform when searching for replacements to level decreases on the
%% non-tree edges.

\myparagraph{Handling Batches of Deletions}
Another challenge is that processing a batch of deletions can shatter
a component into multiple disconnected pieces. Since the HDT algorithm
deletes at most a single tree edge per deletion operation, it handles
exactly two disconnected pieces per level. In contrast, since we
delete batches of edges in our \batchdynamic{} algorithm, we may have
many disconnected pieces at a given level, and must search for
replacement edges reconnecting these pieces. Our algorithm searches
for a replacement edge from each piece that is small enough to be
pushed down to the next lower level.

%% Unlike the HDT
%% algorithm, which only deletes a single tree edge per deletion
%% operation
%%
%% Unlike the HDT
%% algorithm, which only searches the smaller of the two pieces at a
%% given level for a replacement edge, a \batchdynamic{} algorithm may
%% have to search for replacement edges from \emph{multiple pieces} at a given level
%% %\guy{for what?}
%% , since many replacement tree edges must be discovered.
%% %\guy{too many multiples..they are multiplying.}

Each round of both of our algorithms can be viewed as calling an
oracle which returns a set of replacement edges incident on the
disconnected pieces that we are trying to reconnect.
%\guy{again, for what?}
Unlike in the HDT algorithm which terminates once it finds any
replacement edge, the edges returned by the oracle may not fully
restore the connectivity of the original component. In particular, the
replacement edges that we find may contain multiple edges going
between the same pieces (like in \Boruvka{}'s algorithm) or may
contain cycles, which must be dealt with since each level of the data
structure represents a forest.
%\guy{I'm not sure parallel edges is clear}
%(as in \Boruvka{}'s algorithm) or cycles\guy{and why is this bad?}.
Our approach to handling these issues is to run a static spanning
forest algorithm on the replacement edges found in this round, and
insert only the spanning forest edges into the ET-tree at the current
level.

Both of our algorithms alternate between a first phase which calls the
oracle to find a set of replacement edges, and a second phase which
determines a set of replacement edges that can be committed as tree
edges. The difference is that our first (simpler) algorithm
(Section~\ref{sec:simple-parallel}) requires $O(\log^2 n)$ oracle queries per
level, whereas the second algorithm
(Section~\ref{sec:interleaved-parallel}) only requires $O(\log n)$
oracle queries due to a more careful doubling technique.

\myparagraph{Dynamic trees supporting batching}
Another obstacle to improving on the bounds of the HDT algorithm is
that the classic ET-tree data structure performs links and cuts one at
a time. To achieve good batch bounds for forest operations, we use a
recently developed solution to the \batchparallel{} forest
connectivity problem by Tseng et al.~\cite{tseng2018batch}. Their data
structure, which we refer to as a \batchparallel{} ET-tree processes a
set of $k$ links, cuts, or connectivity queries in $O(k \log (1 +
n/k))$ work and $O(\log n)$ depth. We extend the data structure to
supports operations such as fetching the first $l$ non-tree edges in
the tree efficiently.

\section{Preliminaries}\label{sec:prelims}

\myparagraph{Model}
In this paper we analyze our algorithms in the work-depth model using
fork-join style parallelism. Specifically, we use a particular work-depth
model called the \emph{Threaded Random Access Machine} (\mpram{}),
which is closely related to the \pram{} but more closely models
current machines and programming paradigms that are asynchronous and
support dynamic forking. The model can work-efficiently cross-simulate
a \crcwpram{}, equipped with the same atomic instructions, and is
therefore essentially equivalent to the classic \crcwpram{} model. We
formally define the model and provide more details about the
simulations in Appendix~\ref{sec:app-model},
and refer the interested reader to~\cite{blelloch18notes} for full
details.

% Appendix~\ref{sec:app-model}

Our algorithms are designed using \emph{nested fork-join parallelism}
in which a procedure can $\emph{fork}$ off another procedure call to
run in parallel and then wait for forked calls to complete with a
$\emph{join}$ synchronization~\cite{Blelloch96}.
Our efficiency bounds are stated in terms of work and depth, where
\defn{work} is the total number of vertices in the \process{} DAG and
where \defn{depth} (\defn{span}) is the length of the longest path in
the DAG~\cite{Blelloch96}.

%% Our algorithms only require the compare-and-swap atomic primitive,
%% which is widely available on modern multicores.  A
%% $\textproc{compare-and-swap}(\&x, o, n)$ (CAS) instruction takes a memory
%% location $x$ and atomically updates the value at location $x$ to $n$
%% if the value is currently $o$, returning $\codevar{true}$ if it
%% succeeds and $\codevar{false}$ otherwise.

\myparagraph{Parallel Primitives}
The following parallel procedures are used throughout the paper.
A \defn{semisort} takes an input array of
elements, where each element has an associated key and reorders the
elements so that elements with equal keys are contiguous, but elements
with different keys are not necessarily ordered.  The purpose is to
collect equal keys together, rather than sort them. Semisorting a
sequence of length $n$ can be performed in $O(n)$ expected work and
$O(\log n)$ depth w.h.p. assuming access to a uniformly
random hash function mapping keys to integers in the range $[1,
n^{O(1)}]$~\cite{Reif99, gu2015top}.

A \defn{parallel dictionary} data structure supports batch insertion,
batch deletion, and batch lookups of elements from some universe with
hashing.  Gil et al.\ describe a parallel dictionary that uses linear
space and achieves $O(k)$ work and $O(\log^* k)$ depth w.h.p.\ for a
batch of $k$ operations~\cite{gil1991towards}.

The \defn{pack} operation takes an $n$-length sequence $A$ and an
$n$-length sequence $B$ of booleans as input. The output is a sequence
$A'$ of all the elements $a \in A$ such that the corresponding
entry in $B$ is $\codevar{true}$. The elements of $A'$ appear in the
same order that they appear in $A$. Packing can be easily implemented
in $O(n)$ work and $O(\log n)$ depth~\cite{JaJa92}.

\myparagraph{Useful Lemmas}
The following lemmas are useful for analyzing the work bounds of our
parallel algorithms. We provide proofs in Appendix~\ref{sec:additional-proofs}.
\begin{restatable}{lemma}{componentbounds}\label{lem:component_bounds}
	Let $n_1, n_2, ..., n_c$ and $k_1, k_2, ..., k_c$ be sequences of non-negative integers such that $\sum k_i = k$, and $\sum n_i = n$. Then
	\begin{equation}
		\sum_{i=1}^c k_i \log\left(1 + \frac{n_i}{k_i}\right) \leq k \log\left(1 + \frac{n}{k}\right).
	\end{equation}
\end{restatable}

\begin{restatable}{lemma}{batchboundisrootdominated}\label{lem:batch_bound_is_root_dominated}
	For any non-negative integers $n$ and $r$,
	\begin{equation}
	\sum_{w=0}^{r} 2^{w} \log \left( 1 + \frac{n}{2^{w}} \right) = O\left( 2^{r}\log\left(1 + \frac{n}{2^{r}} \right) \right).
	\end{equation}
\end{restatable}

\begin{restatable}{lemma}{batchboundisincreasing}\label{lem:batch_bound_is_increasing}
	For any $n \geq 1$, the function $x \log\left(1 + \frac{n}{x}\right)$ is strictly increasing with respect to $x$ for $x \geq 1$.
\end{restatable}

\subsection{\batchdynamiccap{} Trees}\label{subsec:etts}
The \batchdynamic{} trees problem is to represent a forest as it
undergoes batches of links, cuts, and connectivity queries.
A \defn{link} operation inserts an edge connecting two trees in the
forest. A \defn{cut} deletes an edge from the forest, breaking one
tree into two trees. A \defn{connectivity} query returns whether two
vertices are connected by a path (in the same tree) in the forest.
The interface is formally defined as follows:

\myparagraph{\batchdynamiccap{} Trees Interface}
\begin{itemize}[leftmargin=*]
  \item \textbf{$\textproc{BatchLink}(\{(u_1, v_1), \ldots, (u_k,
      v_k)\})$} takes a sequence of edges and adds them to the graph
    $G$. The input edges must not create a cycle in $G$.

  \item \textbf{$\textproc{BatchCut}(\{(u_1, v_1), \ldots, (u_k,
      v_k)\})$} takes a sequence of edges and removes them from the
    graph $G$.

  \item \textbf{$\textproc{BatchConnected}(\{(u_1,v_1), \ldots, (u_k,
      v_k)\})$} takes a sequence of tuples representing queries. The
    output is a sequence where the $i$-th entry returns whether vertices
    $u_i$ and $v_i$ are connected by a path in $G$.

  \item \textbf{$\textproc{BatchFindRepr}(\{(x_1, \ldots, x_k\})$}
    takes a sequence of pointers to tree elements. The output is a
    sequence where the $i$-th entry is the \emph{representative}
    (repr) of the tree in which $x_i$ lives. The representative is
    defined so that $\emph{repr}(u) = \emph{repr}(v)$ if and only if
    $u$ and $v$ are in the same tree. Note that representatives are
    invalidated after the tree is modified.
\end{itemize}

\myparagraph{\batchparallelcap{} Euler Tour Trees}
In this paper we make use of a recently developed, parallel solution
to the \batchdynamic{} trees problem, called a \batchparallel{} Euler
tour tree (\batchparallel{} ET-trees)~\cite{tseng2018batch}. The data
structure represents each ET-tree sequence using a concurrent
skip-list, and reduces bulk link, cut, and query operations to bulk
operations on the concurrent skip-list.
Tseng et al.~\cite{tseng2018batch} prove the following theorem on the
efficiency of the \batchparallel{} ET-tree:
\begin{theorem}\label{thm:ett-bounds}
A batch of $k$ links, $k$ cuts, $k$ connectivity queries, or $k$
representative queries over an $n$-vertex forest can be processed in $O(k \log (1 + n/k))$ expected work and $O(\log n)$ depth with high probability.
\end{theorem}

\noindent The trees also support augmentation with an associative and
commutative function $f: D^2 \to D$ with values from $D$ assigned to
vertices and edges of the forest. The goal is to compute $f$ over
subtrees of the represented forest. The interface can be easily
extended with the following \batchparallel{} primitives for updating
and querying augmented values.

Appendix~\ref{sec:extra-tree-ops} contains information about
additional tree operations that are needed to efficiently implement
our algorithms.

% Appendix~\ref{sec:extra-tree-ops} contains information about
% additional tree operations that are needed to efficiently implement
% our algorithms.

\subsection{The sequential (HDT) algorithm}\label{sec:hdt}

Our parallel algorithm is based on the sequential algorithm of Holm,
De Lichtenberg, and Thorup \cite{holm2001poly}, which we refer to as
the HDT algorithm. The HDT algorithm assigns to each edge in the
graph, an integer \textit{level} from $1$ to $\log n$. The levels
correspond to sequence of subgraphs $G_1 \subset G_2 \subset ...
\subset G_{\log n} = G$, such that $G_i$ contains all edges with
level at most $i$. The algorithm also maintains a spanning forest
$F_i$ of each $G_i$ such that $F_1 \subset F_2 \subset ... \subset
F_{\log n}$. Each forest is maintained using a set of augmented
ET-trees which we describe shortly. Throughout the algorithm, the
following invariants are maintained.

\begin{invariant}\label{inv:component_sizes}
$\forall i = 1 ...  \log n$, the connected components of $G_i$ have
size at most $2^i$.
\end{invariant}

\begin{invariant}\label{inv:minimum_forest}
$F_{\log n}$ is a minimum spanning forest where the weight of each
edge is its level.
% weighted with the
% levels of the edges.
\end{invariant}

\myparagraph{Connectivity Queries}
To perform a connectivity query in $G$, it suffices to query
$F_{\log n}$, which takes $O(\log n)$ time by querying for the root
of each Euler tour tree and returning whether the roots are equal. We
note that in \cite{holm2001poly}, a query time of
$O\left(\log n /\log\log n\right)$ is achieved by storing the
Euler tour of $F_{\log n}$ in a B-tree with branching factor
$\log n$.

\myparagraph{Inserting an Edge}
An edge insertion is handled by assigning the edge to level $\log n$.
If the edge connects two currently disconnected components, then it is
added to $F_{\log n}$.

\myparagraph{Deleting an Edge}
Deletion is the most interesting part of the algorithm. If the deleted
edge is not in the spanning forest $F_{\log n}$, the algorithm
removes the edge and does nothing to $F_{\log n}$ as the connectivity
structure of the graph is unchanged.  Otherwise, the component
containing the edge is split into two. The goal is to find a
\defn{replacement edge}, that is, an edge crossing the split
component.

If the deleted edge had level $i$, then the \emph{smaller} of the
two resulting components is searched starting at level $i$ in order to
locate a replacement edge. Before searching this component, all tree
edges whose level is equal to $i$ have their level decremented by one.
As the smaller of the split components at
level $i$ has size $\leq 2^{i-1}$, pushing the entire component to
level $i-1$ does not violate Invariant~\ref{inv:component_sizes}.
Next, the non-tree edges at level $i$ are considered one at a time as
possible replacement edges. Each time the algorithm examines an
edge that is not a replacement edge, it decreases the level of the
edge by one. If no replacement is found, it moves up to the next level
and repeats. Note that because the algorithm first pushes all tree
edges to level $i-1$, any subsequent non-tree edges that may be pushed
from level $i$ to level $i-1$ will not violate
Invariant~\ref{inv:minimum_forest}.

\myparagraph{Implementation and Cost}
To efficiently search for replacement edges, the ET-trees are
augmented with two additional pieces of information. The first
augmentation is to maintain the number of non-tree edges whose level
equals the level of the tree.  The second augmentation maintains the
number of tree-edges whose level is equal to the level of the tree.

Using these augmentations, each successive non-tree edge (or tree
edge) whose level is equal to the level of the tree can be found in
$O(\log n)$ time. Furthermore, checking whether the edge is a
replacement edge can be done in $O(\log n)$ time. Lastly, the cost of
pushing an edge that is not a replacement edge to the lower level is
$O(\log n)$, since it corresponds to inserting the edge into an
adjacency structure and updating the augmented values. Since each edge
can be processed at most once per level, paying a cost of $O(\log n)$,
and there are $\log n$ levels, the overall amortized cost per edge is
$O(\log^2 n)$.

\section{A Parallel Algorithm}\label{sec:simple-parallel}

In this section, we give a simple parallel \batchdynamic{}
connectivity algorithm based on the HDT algorithm. The underlying
invariants maintained by our parallel algorithm are identical to the
sequential HDT algorithm: we maintain $\log n$ levels of spanning
forests subject to Invariants~\ref{inv:component_sizes}
and~\ref{inv:minimum_forest}.  The main challenge, and where our
algorithm departs from the HDT algorithm is in how we search for
replacement edges in parallel, and how we search multiple components
in parallel. We show by a charging argument that this parallel
algorithm is work-efficient with respect to the HDT algorithm---it
performs $O(\log^2 n)$ amortized work per edge insertion or deletion.
Furthermore, we show that the depth of this algorithm is $O(\log^4
n)$. Although these bounds are subsumed by the improved parallel
algorithm we describe in Section~\ref{sec:interleaved-parallel}, the
parallel algorithm in this section is useful to illustrate the main
ideas in this paper.

%% In
%% Section~\ref{sec:interleaved-parallel} we derive an improved parallel
%% \batchdynamic{} algorithm that achieves better work and depth bounds.

%\myparagraph{Definitions}

% we assign each edge a level from $1$
% to $\log n$. $G_i$ denotes the subgraph consisting of edges with level
% at most $i$. We maintain spanning forests $F_1, F_2, ..., F_{\log n}$
% subject to Invariant \ref{inv:component_sizes} and Invariant
% \ref{inv:minimum_forest}.

\myparagraph{Data Structures}
Each spanning forest, $F_i$, is represented using a set of \batchparallel{}
ET-trees~\cite{tseng2018batch}. We represent the edges of the
graph in a parallel dictionary $E_{D}$ for convenience (see
Section~\ref{sec:prelims}). We also store an adjacency array,
$A_{i}[u]$, at each level $i$, and for each vertex $u$ to store the
tree and non-tree edges incident on $u$ with level $i$. Note that
tree and non-tree edges are stored separately so that they can be
accessed separately. The adjacency arrays support batch insertion
and deletion of edges, as well as the ability to fetch a batch of
edges of a desired size. These operations have the following cost bounds.

\begin{restatable}{lemma}{adjacencyops}\label{lem:adjacencyops}
	\textproc{InsertEdges},
	\textproc{DeleteEdges}, and \textproc{FetchEdges}
	can be implemented in $O(1)$ amortized work per edge and in $O(\log n)$
	depth.
\end{restatable}

\noindent See Appendix~\ref{sec:datastructures} for proofs and full details on the adjacency
data structure.

% Refer to Appendix~\ref{sec:datastructures} for proofs and full
% details on the adjacency data structure.

\subsection{Connectivity Queries}\label{subsec:batch-query}

As in the sequential algorithm, a connectivity query can be answered
by simply performing a query on $F_{\log n}$.
Algorithm~\ref{alg:batch-query} gives pseudocode for the batch
connectivity algorithm. The bound we achieve follows from the batch
bounds on \batchparallel{} ET-trees.

\begin{algorithm}[H]
\caption{The batch query algorithm}\label{alg:batch-query}
\small
\begin{algorithmic}[1]
\Procedure{BatchQuery}{$\{ (u_1, v_1), (u_2, v_2), ..., (u_k, v_k)  \}$}
  \State \textbf{return} $F_{{\log n}}$.\textsc{BatchQuery}($\{ (u_1, v_1), (u_2, v_2), ..., (u_k, v_k)  \}$)
\EndProcedure
\end{algorithmic}
\end{algorithm}

\begin{theorem}\label{thm:connectivity-queries}
A batch of $k$ connectivity queries can be processed in $O\left(k
  \log\left(1 + \frac{n}{k}\right)\right)$ expected work and $O(\log n)$ depth w.h.p.
\end{theorem}
\begin{proof}
  Follows from Theorem~\ref{thm:ett-bounds}.
\end{proof}

\subsection{Inserting Batches of Edges}\label{subsec:batch-insert}
To perform a batch insertion, we first determine a set of edges in the
batch that increase the connectivity of the graph. To do so, we treat
each current connected component of the graph as a vertex, and build a
spanning forest of the edges being inserted over this contracted
graph. The edges in the resulting spanning forest are then inserted
into the topmost level in parallel.

\begin{algorithm}[H]
\caption{The batch insertion algorithm}
\label{alg:batch-insertion}
\small
\begin{algorithmic}[1]
\Procedure{BatchInsert}{ $U = \{(u_1, v_1), \ldots, (u_k, v_k)  \}$ }
  \State For all $e_i \in U$, set $l(e_i) \assign \log n$ in parallel\label{bi:setlevel}
  \State Update $A_{\log n}[u]$ for edges incident on $u$ \label{bi:insertedges}
%  \State $U' \assign$ Apply $F_{\log n}$.\textsc{BatchFindRepr} to the endpoints of $e \in U$\label{bi:batchfindflogn}
  \State $R \assign \{ (F_{\log n}.\textsc{FindRepr}(u), F_{\log n}.\textsc{FindRepr}(u))\ |\ (u,v) \in U\}$\label{bi:batchfindflogn}
  \State $T' \assign$ \textsc{SpanningForest}($R$)\label{bi:insertspanningforest}
  \State $T \assign$ edges in $U$ corresponding to $T'$
  \State Promote edges in $T$ to tree edges\label{bi:insertpromote}
  \State $F_{\log n}$.\textsc{BatchInsert}(T)\label{bi:inserttoflogn}
\EndProcedure
\end{algorithmic}
\end{algorithm}

\noindent Algorithm~\ref{alg:batch-insertion} gives pseudocode for the batch
insertion algorithm. We assume that the edges given as input in $U$
are not present in the graph. Each vertex $u$ that receives an
updated edge inserts its edges into $A_{\log n}[u]$
(Line~\ref{bi:insertedges}). This step can be implemented by first
running a semisort to collect all edges incident on
$u$.

The last step is to insert edges that increase the connectivity of
the graph as tree edges
(Lines~\ref{bi:batchfindflogn}--\ref{bi:inserttoflogn}). The
algorithm starts by computing the representatives for each edge
(Line~\ref{bi:batchfindflogn}). The output is an array of edges,
$R$, which maps each original $(u,v)$ edge in $U$ to
$(\textsc{FindRepr}(u), \textsc{FindRepr}(v))$ (note that these calls
can be batched using \textsc{BatchFindRepr}). Next, it computes a
spanning forest over the tree edges
(Line~\ref{bi:insertspanningforest}).  Finally, the algorithm promotes
the corresponding edges in $U$ to tree edges. This is done by updating
the appropriate adjacency lists and inserting them into $F_{\log n}$
(Lines~\ref{bi:insertpromote}--\ref{bi:inserttoflogn}).

\begin{theorem}\label{thm:insertions}
A batch of $k$ edge insertions can be processed in $O\left(k
  \log\left(1 + \frac{n}{k}\right)\right)$ expected work and
$O(\log n)$ depth w.h.p.
\end{theorem}
\begin{proof}
  Lines~\ref{bi:setlevel}--\ref{bi:insertedges} cost $O(k)$ work and
  $O(\log k)$ depth w.h.p.\ using our bounds for updating $A$ (see
  Lemma~\ref{lem:adjacencyops}).
The find representative queries (Line~\ref{bi:batchfindflogn}) can be
implemented using a \textsc{BatchFindRepr} call, which
costs $O\left(k \log\left(1 + \frac{n}{k}\right)\right)$ expected work
and $O(\log n)$ depth w.h.p.\ by Theorem~\ref{thm:ett-bounds}.
Computing a spanning forest (Line~\ref{bi:insertspanningforest}) can
be done in $O(k)$ expected work and $O(\log k)$ depth w.h.p.\ using
Gazit's connectivity algorithm \cite{gazit1991optimal}. Finally,
updating the adjacency lists and inserting the spanning forest edges
into $F_{\log n}$ costs $O\left(k \log\left(1 +
    \frac{n}{k}\right)\right)$ expected work and $O(\log n)$ depth
w.h.p. (Lines~\ref{bi:insertpromote}--\ref{bi:inserttoflogn}).
\end{proof}

\subsection{Deleting Batches of Edges}
As in the sequential HDT algorithm, searching for replacement edges
after deleting a batch of tree edges is the most interesting part of
our parallel algorithm. A natural idea for parallelizing the HDT
algorithm is to simply scan all non-tree edges incident on each
disconnected component in parallel.
Although this approach has low depth per level, it may examine a huge
number of candidate edges, but only push down a few non-replacement
edges. In general, it is unable to amortize the work performed
checking all canidates edges at a level to the edges that experience
level decreases.  To amortize the work properly while also searching
the edges in parallel we must perform a more careful exploration of
the non-tree edges. Our approach is to use a \emph{doubling}
technique, in which we geometrically increase the number of non-tree
edges explored as long as we have not yet found a replacement edge. We
show that using the doubling technique, the work performed (and number
of non-tree edges explored) is dominated by the work of the last
phase, when we either find a replacement edge, or run out of non-tree
edges. Our amortized work-bounds follow by a per-edge charging
argument, as in the analysis of the HDT algorithm.

\myparagraph{The Deletion Algorithm}
Algorithm~\ref{alg:batch-deletion} shows the pseudocode for our
parallel batch deletion algorithm. As with the batch insertion
algorithm, we assume that each edge is present in $U$ in both
directions. Given a batch of $k$ edge deletions, the algorithm first
deletes the
given edges from their respective adjacency lists in parallel
(Line~\ref{bd:deletea}). It then filters out the tree edges
(Line~\ref{bd:identifytree}) and
deletes each tree edge $e$ from $F_{i} \ldots, F_{\log n}$, where $i$
is the level of $e$ (Line~\ref{bd:deletetree}). Next, it computes $C$,
a set of \emph{components} (representatives) from the deleted tree
edges (Line~\ref{bd:findcomps}). For each deleted tree edge, $e$, the
algorithm includes the representatives of both endpoints in the forest
at $l(e)$, which must be in different components as $e$ is a deleted
tree edge. Finally, the algorithm loops over the levels, starting at
the lowest level where a tree edge was deleted
(Line~\ref{bd:looplevels}), and calls \parlevelsearch{} at each level.
Each call to \parlevelsearch{} takes $i$, the level to search, $C$,
the current set of disconnected components, and $S$, an initially
empty set of replacement edges that the algorithm discovers over the
course of the searches
(Line~\ref{bd:levelsearch})

\begin{algorithm}[]
\caption{The batch deletion algorithm}
\label{alg:batch-deletion}
\small
\begin{algorithmic}[1]
\Procedure{BatchDeletion}{$U = \{e_1, \ldots, e_k\}$}
  \State Delete $e \in U$ from $A_{0}, \ldots, A_{\log n}$\label{bd:deletea}
  \State $T \assign \{ e \in U\ |\ e \in F_{\log n}\}$\label{bd:identifytree} \algcomment{Tree edges to delete}
  \State Delete $e \in T$ from $F_{0}, \ldots, F_{\log n}$\label{bd:deletetree}
  \State $C \assign \cup_{e=(u,v) \in T} (F_{l(e)}.\textsc{FindRepr}(u), F_{l(e)}.\textsc{FindRepr}(v))$\label{bd:findcomps}
  \State $S \assign \emptyset$\label{bd:initS}
  \For{$i \in [min_{l} \assign \min_{e \in T}, \log n]$}\label{bd:looplevels}
    \State $(C, S) \assign$ \parlevelsearch{}($i, C, S$)\label{bd:levelsearch}
  \EndFor
\EndProcedure
\end{algorithmic}
\end{algorithm}

% sls: simple level search
\begin{algorithm}[h!]
\caption{The parallel level search algorithm} \label{alg:simple-level-search}
\small
\begin{algorithmic}[1]
\Procedure{ComponentSearch}{$i, c$}\label{cs:start}
  \State $w \assign 1,\ w_{\max} \assign c.\textsc{NumNonTreeEdges}$\label{cs:init}
  \While {$w \leq w_{\max}$}\label{cs:whilestart}
    \State $w \assign \min(w, w_{\max})$\label{cs:assignw}
    \State $E_c \assign$ First $w$ non-tree edges in $c$\label{cs:assignnontree}
    \State Push all non-replacement edges in $E_c$ to level $i-1$\label{cs:pushnonrepl}
    \If {$E_c$ contains a replacement edge}\label{cs:ifcheck}
      \State \algorithmicreturn{} $\{ r \}$, where $r$ is any
      replacement edge in $E_c$\label{cs:returnr}
    \EndIf
    \State $w \assign 2w$\label{cs:updatew}
  \EndWhile\label{cs:whileend}
  \State \algorithmicreturn{} $\emptyset$\label{cs:returnempty}
\EndProcedure\label{cs:end}
\smallskip
\Procedure{ParallelLevelSearch}{$i$, $L = \{c_{1},c_2, \ldots\}$, $S$}
  \State $F_{i}$.\textsc{BatchInsert}($S$)\label{sls:batchinserts}
  \State $C \assign c \in L$ with size $\leq 2^{i-1}$\label{sls:smallcomponent}
  \State $D \assign c \in L$ with size $> 2^{i-1}$\label{sls:largecomponent}
  \While{$|C| > 0$}\label{sls:mainloopstart}
    \State Push level $i$ tree edges of components in $C$ to level $i-1$\label{sls:pushtreelower}
    \State $R \assign \cup_{c \in C}\ \textsc{ComponentSearch}(i, c)$\label{sls:findreplacements} \algcomment{In parallel}

    \State $R'\assign \{ (F_{i}.\textsc{FindRepr}(u), F_{i}.\textsc{FindRepr}(v))\ |\ (u, v) \in R\}$\label{sls:applyfindrep}

    \State $T' \assign$ \textsc{SpanningForest}$(R')$\label{sls:spanningforest}
    \State $T \assign$ Edges in $R$ corresponding to edges in $T'$\label{sls:newtreeedges}
    \State Promote edges in $T$ to tree edges\label{sls:promotefi}
    \State $F_{i}$.\textsc{BatchInsert}(T)\label{sls:batchinsertfi}
    \State $S \assign S \cup T$\label{sls:updateS}
%    \State $C \assign F_{i}.\textsc{BatchFindRepr}(C)$, and remove duplicates\label{sls:updateallcomponents}
    \State $C \assign \{ F_{i}.\textsc{Repr}(c)\ |\ c \in C\}$ \label{sls:updateallcomponents}
    \State $Q \assign$ $\{ c \in C$ with no non-tree edges, or size $> 2^{i-1} \}$ \label{sls:setQ}
    \State $D \assign D \cup Q$\label{sls:updateD}
    \State $C \assign C \setminus Q$\label{sls:updateC}
  \EndWhile\label{sls:mainloopend}
  \State \algorithmicreturn{} $(D, S)$\label{sls:return}
\EndProcedure
\end{algorithmic}
\end{algorithm}

\myparagraph{Searching a Level in Parallel}
The bulk of the work done by the deletion algorithm is performed by
Algorithm~\ref{alg:simple-level-search}, which implements a subroutine
that searches the disconnected components at a given level of the data
structure in parallel. The input to \parlevelsearch{} is an integer
$i$, the level to search, a set of representatives of the disconnected
components, $L$, and the set of replacement spanning forest edges that
were found in levels lower than $i$, $S$. The output of
\parlevelsearch{} is the set of components that are still disconnected
after considering the non-tree edges at this level, and the set of
replacement spanning forest edges found so far.

\parlevelsearch{} starts by inserting the new spanning forest edges in $S$ into
$F_{i}$ (Line~\ref{sls:batchinserts}). Next, it computes $C$ and $D$, which are the
components that are active and inactive at this level, respectively
(Lines~\ref{sls:smallcomponent}--\ref{sls:largecomponent}). The main loop
of the algorithm (Lines~\ref{sls:mainloopstart}--\ref{sls:mainloopend})
operates in a number of \defn{rounds}.
Each round first pushes down all tree edges at level $i$ of every
active component. It then finds a single
replacement edge incident to each active component, searching the
active components in parallel, pushing any non-replacement edge to
level $i-1$. It then promotes a maximal acyclic subset of the
replacement edges found in this round to tree edges, and proceeds to
the next round. The rounds terminate once all components at this level
are deactivated by either becoming too large to search at this level,
or because the algorithm finished examining all non-tree edges
incident to the component at this level.

The main loop (Lines~\ref{sls:mainloopstart}--\ref{sls:mainloopend})
works as follows.
The algorithm first pushes any level $i$ tree edges in an active
component down to level $i-1$.  The active components in $C$ have size
at most $2^{i-1}$, meaning that any tree edges they have at level $i$
can be pushed to level $i-1$ (Line~\ref{sls:pushtreelower}) without
violating Invariant~\ref{inv:component_sizes}.
Next, the algorithm searches each active component for a replacement
edge in parallel by calling the \textsc{ComponentSearch} procedure in
parallel over all components (Line~\ref{sls:findreplacements}).  This
procedure either returns an empty set if there are no replacement
edges incident to the component, or a set containing a single
replacement edge. Next, the algorithm maps the replacement edge
endpoints to their current component's representatives by calling
\textsc{FindRepr} on each endpoint (Line~\ref{sls:applyfindrep}). It then computes
a spanning forest over these replacement edges
(Line~\ref{sls:spanningforest}) and maps the edges included in the
spanning forest back to their original endpoints ids
(Line~\ref{sls:newtreeedges}). Observe that the edges in $T$
constitute a maximal acyclic subset of replacement edges of $R$ in
$F_i$. The algorithm therefore promotes the edges in $T$ to tree edges
(Lines~\ref{sls:promotefi}--~\ref{sls:batchinsertfi}).  Note that the new tree edges are
not immediately inserted into all higher level spanning trees.
Instead, the edges are buffered by adding them to $S$
(Line~\ref{sls:updateS}) so that they will be inserted when the higher
level is reached in the search. Finally, the algorithm updates the set
of components by computing their representatives on the updated $F_i$
(Line~\ref{sls:updateallcomponents}), and filtering out any components
which have no remaining non-tree edges, or become larger than
$2^{i-1}$ (i.e., become unsearchable at this level) into $D$
(Lines~\ref{sls:setQ}--\ref{sls:updateC}).

We now describe the \textsc{ComponentSearch} procedure
(Lines~\ref{cs:start}--\ref{cs:end}).  The search consists of a number
of \defn{phases}, where the $i$'th phase
searches the first $2^{i}$ non-tree edges, or all of the non-tree
edges if $2^{i}$ is larger than the number of non-tree edges in $c$.
The search terminates either once a replacement edge incident to $c$
is found (Line~\ref{cs:ifcheck}), or once the algorithm unsuccessfuly
examines all non-tree edges incident to $c$
(Line~\ref{cs:whilestart}). Initially $w$, the search size, is set to
$1$ (Line~\ref{cs:assignw}). On each phase, the algorithm retrieves
the first $w$ many non-tree edges, $E_c$
(Line~\ref{cs:assignnontree}). It pushes all non-tree edges that are
not replacements to level $i-1$ (Line~\ref{cs:pushnonrepl}).
It then checks whether any of the edges
in $E_c$ are a replacement edge, and if so, returns one of the
replacement edges in $E_c$ (Line~\ref{cs:returnr}). Note that checking
whether an edge is a replacement edge is done using
\textsc{BatchFindRepr}. Otherwise, if no replacement edge was found it
doubles $w$ (Line~\ref{cs:updatew}) and continues.

\myparagraph{Cost Bounds}
We now prove that our parallel algorithm has low depth, and is
work-efficient with respect to the sequential HDT algorithm. For
simplicity, we assume that we start with no edges in a graph on $n$
vertices.

%%\begin{lemma}
%%The number of phases of doubling required to find the first
%%replacement edge incident on a component is $O(\log n)$.
%%\end{lemma}\label{lem:num-phases}
%%\begin{proof}
%%The algorithm doubles the number of edges searched in each phase.
%%Therefore, after $\log_{2} m = O(\log n)$ phases, all non-tree edges
%%incident on the component will be searched.
%%\end{proof}
%%
%%\begin{lemma}
%%The maximum number of rounds performed by the algorithm at a given
%%level is at most $O(\log n)$.
%%\end{lemma}\label{lem:num-rounds}
%%\begin{proof}
%%Each round of the algorithm either gets rid of all components, or
%%finds a replacement edge for each remaining component. In the worst
%%case, the edges found for each component pair the components off,
%%leaving us with half as many components in the subsequent round. As we
%%lose a constant fraction of the components per round, the algorithm
%%takes $O(\log n)$ rounds.
%%\end{proof}

\begin{theorem}\label{thm:deletion-depth}
A batch of $k$ edge deletions can be processed in $O(\log^4 n)$ depth w.h.p.
\end{theorem}
\begin{proof}
%By Lemma~\ref{lem:num-phases}, the number of phases of the algorithm
%in a given level is at most $O(\log n)$. By Lemma~\ref{lem:num-rounds}, the
%number of rounds is $O(\log n)$.

The algorithm doubles the number of edges searched in each phase.
Therefore, after $\log m = O(\log n)$ phases, all non-tree edges
incident on the component will be searched.

In every round, each active component is either deactivated, or has a replacement
edge found. In the worst case, the edges found for each active component pair the components off,
leaving us with half as many active components in the subsequent round. As we
lose a constant fraction of the active components per round, the algorithm
takes $O(\log n)$ rounds.

A given level can therefore perform at most $O(\log^2 n)$ phases.
Each phase consists of fetching,
examining, and pushing down non-tree edges, and hence can be implemented
in $O(\log n)$ depth w.h.p.\ by Lemma~\ref{lem:ettree-fetch-edges},
Theorem~\ref{thm:ett-bounds}, and Lemma~\ref{lem:ettree-decrease-level}.
Therefore, the overall depth for a given level is $O(\log^3 n)$ w.h.p.
As all $\log n$ levels will be processed in the worst case, the
overall depth of the algorithm is $O(\log^4 n)$ w.h.p.
\end{proof}

\noindent We now analyze the work performed by the algorithm. We begin by stating the
following lemmas on the efficiency of the augmented ET-tree operations.

\begin{restatable}{lemma}{ettreefetchedges}\label{lem:ettree-fetch-edges}
	Given some vertex, $v$ in a \batchparallel{} ET-tree, we can fetch
	the first $l$ tree (or non-tree) edges referenced by the augmented
	values in the tree in $O\left(l \log\left(1 +
	\frac{n_c}{l}\right)\right)$ work and $O(\log n)$ depth w.h.p.\ where $n_c$
	is the number of vertices in the ET-tree at the current
	level. Furthermore, removing the edges can be done in the same bounds.
\end{restatable}

\begin{restatable}{lemma}{ettreedecreaselevel}\label{lem:ettree-decrease-level}
	Decreasing the level of $l$ tree (or non-tree) edges in a
	\batchparallel{} ET-tree can be performed in $O\left(l \log\left(1 +
	\frac{n_c}{l}\right)\right)$ expected work and $O(\log n)$ depth
	w.h.p.\ where $n_c$ is the number of nodes in the ET-tree at the
	current level.
\end{restatable}

\noindent The proofs of these lemmas are provided in Appendix~\ref{sec:extra-tree-ops}. We now have the tool to analyze the work performed by
batch deletion.

\begin{lemma}\label{lem:batch_delete_exclude_search}
  The work performed by \textsc{BatchDeletion} excluding the calls to
  \textsc{ParallelLevelSearch} is
	\begin{equation}
	O\left(k\log n\log\left(1 + \frac{n}{k}\right)\right),
	\end{equation}
  in expectation.
\end{lemma}
\begin{proof}
  The edge deletions performed by Line~\ref{bd:deletea} cost $O(k)$ work by
  Lemma~\ref{lem:adjacencyops}.
  Filtering the tree edges (Line~\ref{bd:identifytree}) can be done in
  $O(k)$ work. Deleting the tree edges costs at most
  $O\left(k \log\left(1 + {n}/{k}\right)\right)$ work by
  Lemma~\ref{lem:batch_bound_is_increasing} (Line~\ref{bd:deletetree}).

  Line~\ref{bd:findcomps} perform a \textsc{FindRepr} call for each endpoint of each
  deleted tree edge. These calls can be implemented as a single
  \textsc{BatchFindRepr} call which costs $O\left(k \log\left(1 +
      {n}/{k}\right)\right)$ work in expectation by Theorem~\ref{thm:ett-bounds}. Since in the worst case each tree edge
  must be deleted from $\log n$ levels, the overall cost of this step
  is $O\left(k \log n\log\left(1 + {n}/{k}\right)\right)$ in
  expectation. Summing up the costs for each level proves the lemma.
\end{proof}

\begin{theorem}\label{thm:simple-work}
The expected amortized cost per edge insertion or deletion is $O(\log^2 n)$.
\end{theorem}
\begin{proof}
Algorithm~\ref{alg:batch-deletion} takes as input a batch of $k$ edge
deletions. By Lemma~\ref{lem:batch_delete_exclude_search}, the
expected work performed by \textsc{BatchDeletion} excluding the calls
to \parlevelsearch{} is
\begin{equation}
O\left(k\log n\log\left(1 + \frac{n}{k}\right)\right),
\end{equation}
which is at most $O(k\log^2 n)$ in
expectation. We now consider the cost of the calls to
\parlevelsearch{}. Specifically, we show that the work performed during the calls to
\parlevelsearch{} can either be charged to level decreases on edges,
or is at most $O(k \log n)$ per call in expectation. Since the total number of calls
to \parlevelsearch{} is at most $\log n$, the bounds follow.

First, observe that the number of spanning forest edges we discover, $|S|$,
is at most $k$, since at most $k$ tree edges were deleted initially.
Therefore, the batch insertion on Line~\ref{sls:batchinserts} costs $O(k \log n)$ in expectation
by Theorem~\ref{thm:ett-bounds}.
Similarly, $L$, the number of components that are supplied to
\parlevelsearch{}, is at most $k$. Therefore, the cost of filtering
the components in $L$ based on their size, and checking whether their
representative exists in $F_i$ is at most $O(k \log n)$ in expectation
(Lines~\ref{sls:smallcomponent}--\ref{sls:largecomponent}).

To fetch, examine, and push down $l$ tree or non-tree edges costs
\begin{equation}
O\left(l \log\left(1 + \frac{n}{l}\right)\right),
\end{equation}
work in expectation, by Lemma~\ref{lem:ettree-fetch-edges},
Theorem~\ref{thm:ett-bounds}, and Lemma~\ref{lem:ettree-decrease-level}.
Note that this is at most $O(\log n)$ per edge. In particular, the cost
of retrieving and pushing the tree edges of active components to
level $i-1$ (Line~\ref{cs:pushnonrepl}) is therefore at most $O(\log
n)$ per edge in expectation, which we charge to the corresponding
level decreases.

We now show that all work done while searching for replacement edges
(Lines~\ref{sls:mainloopstart}--\ref{sls:mainloopend}) can be charged to level decreases. Consider an active component,
$c$ in some round. Suppose the algorithm performs $q > 0$ phases before
either the component is exhausted (all incident non-tree edges have been checked),
or a replacement edge is found. First consider the case where it finds a
replacement edge. If $q = 1$, only a single edge was inspected, so then we charge
the $\log n$ work for the round to the edge, which will become a tree edge. Otherwise, it performs
$q-1$ phases which do not produce any replacement edge.

Since phase $w$ inspects $2^w$ edges, it costs $O(2^w \log n)$ work. The total work
over all $q$ phases is therefore
\begin{equation}
  \sum_{w=0}^{q} 2^{w} \log n = O(2^{q} \log n)
\end{equation}
in expectation. However, since no replacement was found during the first $q-1$ phases,
there are at least $2^{q-1} = O(2^q)$ edges that will be pushed down, so we can charge $O(\log n)$
work to each such edge to pay for this. In the other case, $q$ phases run without finding a replacement edge.
In this case, all edges inspected are pushed down, and hence each assumes
a cost of $O(\log n)$ in expectation.

Now, we argue that the work done while processing the replacement
edges is $O(k \log n)$ in expectation over all rounds. Since $k$ edges
were deleted, the algorithm discovers at most $k$ replacement edges.
We charge the work in these steps to the replacement edges that we
find. Let $k'$ be the number of replacement edges that we find.
Filtering the edges, and computing a spanning forest all costs $O(k')$
work. Promoting the edges to tree edges (inserting them into $F_i$ and
updating the adjacency lists) costs $O(k' \log n)$ work in
expectation. Finally, updating the components costs $O(k' \log n)$
work in expectation, which we can charge to either the component, if
it is removed from $C$ in this round, or to the replacement edge that
it finds, which is promoted to a tree edge.  Since the algorithm can
find at most $k$ replacement edges, the cost per level is $O(k\log n)$
in expectation for these steps as necessary.

In total, on each level the algorithm performs $O(k \log n)$ expected work that is not charged to a
level decrease. Summing over $\log n$ levels, this yields an
amortized cost of $O(\log^2 n)$ expected work per edge deletion.
Finally, since the level of an edge can decrease at most $\log n$
times, and an edge is charged $O(\log n)$ expected work each time its level is
decreased, the expected amortized cost per edge insertion is $O(\log^2 n)$.

\end{proof}

% ils: interleaved level search
\begin{algorithm}[h!]
\caption{The interleaved level search algorithm} \label{alg:interleaved-level-search}
\small
\begin{algorithmic}[1]
\Procedure{ComponentSearch}{$i, c, s$}\label{fr:start}
  \State $w_{\max} \assign c.\textsc{NumNonTreeEdges}$ \label{fr:wmax}
  \State $w \assign \min(s, w_{\max})$\label{fr:wdef}
  \State $E_c \assign$ First $w$ non-tree edges in $c$ \label{fr:fetchedges}
%  \State $\{ R \assign$ All replacement edges in $E_c \}$ \label{fr:repledges}
  \State \algorithmicreturn{} $\{$All replacement edges in $E_c\}$\label{fr:returnr}
\EndProcedure\label{fr:end}
\smallskip
\Procedure{PushEdges}{$i, c, s, M$}\label{fetp:start}
  \State $w_{\max} \assign c.\textsc{NumNonTreeEdges}$ \label{fetp:wmax}
  \State $w \assign \min(s, w_{\max})$\label{fetp:wdef}
  \State $E_c \assign \{$First $w$ non-tree edges in $c\}$ \label{fetp:fetchedges}
  \If {$M[c].\textsc{size} \leq 2^{i-1}$ {\bf and} $w < w_{\max}$ } \label{fetp:pushingcondition}
    \State Remove edges in $E_c$ from level $i$ \label{fetp:deletefromleveli}
    \State \algorithmicreturn{} $E_c$\label{fetp:returnec}
  \EndIf
  \State \algorithmicreturn{} $\emptyset$ \label{fetp:returnempty}
\EndProcedure\label{fetp:end}
\smallskip
\Procedure{InterleavedLevelSearch}{$i$, $L = \{c_{1}, c_2, \ldots \}$, $S$}
  \State $F_{i}$.\textsc{BatchInsert}($S$)\label{ils:pushs}
  \State $C \assign c \in L$ with size $\leq 2^{i-1}$
  \State $D \assign c \in L$ with size $> 2^{i-1}$
  \State Push level $i$ tree edges of all components in $C$ to level $i-1$ \label{ils:optreepush}
  \State $r \assign 0,\ T \assign \emptyset,\ E_P \assign \emptyset$ \label{ils:opdefs}
  \State $M \assign \{ c \rightarrow c\ |\ c \in C\}$ %\Comment{Map of components to representatives}
  \While{$|C| > 0$}\label{ils:loopstart}
    \State $w \assign 2^{r}$
    \State $R \assign \cup_{c \in C}\ \textsc{ComponentSearch}(i, c, w)$\label{ils:findreplacements} \algcomment{In parallel}

    \State $R'\assign \{ (F_{i}.\textsc{FindRepr}(u), F_{i}.\textsc{FindRepr}(v))\ |\ (u, v) \in R\}$
    %\State $C_{r} \assign C \cup \{$components affected by $R \}$\label{ils:vertexsetcomp}
    \State $T'_{r} \assign$ \textsc{SpanningForest}$(R')$ \label{ils:sfcomp}
    \State $T_{r} \assign$ Edges in $R$ corresponding to edges in $T'_{r}$
    \State $T \assign T \cup T_{r}$ \label{ils:addtoT}
    \State Update $M$, the map of supercomponents and their sizes\label{ils:updateM}
    %\State $M \assign M \cup \{c \rightarrow (M'(M(c)), M'(M(c)).size)\ |\ c \in C_{r}\}$ \label{ils:updateM}

    \State $E_P \assign E_P \cup_{c \in C}\ \textsc{PushEdges}(i, c, w, M)$\label{ils:findedgestopush} \algcomment{In parallel}

    \State $D_r \assign$  $\{c \in C$ with no non-tree edges, or size $> 2^{i-1}\}$  \label{sls:setQ}
    \State $D \assign D\ \cup D_r$ \label{ils:dupdate}
    \State $C \assign C\ \setminus D_r$\label{ils:cupdate}
    \State $r \assign r + 1$ \label{ils:rupdate}
  \EndWhile\label{ils:loopend}

  \State Promote edges in $T$ not in $E_p$ to tree edges at level $i$\label{ils:promotetree}
  \State $F_{i}$.\textsc{BatchInsert}($T$)\label{ils:batchinserttree}
  \State Insert non-tree and tree edges in $E_P$ to level $i - 1$\label{ils:pushdown}
  \State \algorithmicreturn{} $(D, S \cup T)$\label{ils:return}
\EndProcedure
\end{algorithmic}
\end{algorithm}

\section{An Improved Algorithm}\label{sec:interleaved-parallel}

In this section we design an improved version of the parallel algorithm
that performs less work than the algorithm from
Section~\ref{sec:simple-parallel}. Furthermore, the improved algorithm
runs in $O(\log^3 n)$ depth w.h.p., improving on
Algorithm~\ref{alg:simple-level-search} which runs in $O(\log^4 n)$
depth w.h.p.
% The key
% ideas in the improved algorithm involve interleaving the replacement edge
%  search phases with the spanning forest computation, and adding replacement edges to
% the spanning forest more lazily to obtain better batch sizes for ET-tree
% operations.

%% The main idea is to show that if a component $c$ pushes down $p_{c}$
%% edges in total, then the work it performs is $O\left(p_{c} \log\left(1
%%     + \frac{n_{c}}{p_{c}}\right) \right)$, where $n_{c}$ is the number
%% of vertices in $c$. Using this fact, it is to argue that the total
%% work charged to edge insertions is $O\left(p \log\left(1 +
%%     \frac{n}{p}\right) \right)$ where $p$ is the total number of edges
%% pushed down by all components during the search at this level.

\subsection{The Interleaved Deletion Algorithm}

\myparagraph{Overview}
% Don't reset
Algorithm~\ref{alg:interleaved-level-search} is based on
\emph{interleaving} the phases of doubling that search for replacement
edges with the spanning forest computation performed on the
replacement edges. Recall that in
Algorithm~\ref{alg:simple-level-search}, the number of edges examined
in each round \emph{is reset}, and the doubling algorithm must
therefore start with an initial search size of $1$ on the next round.
Because the doubling resets from round to round, the number of phases
per round can be $O(\log n)$ in the worst case, making the total
number of phases per level $O(\log^2 n)$. Instead, the interleaved algorithm avoids
resetting the search size by maintaining a \emph{single},
geometrically increasing search size over all rounds of the search.

% Defer inserting tree edges, search original components
The second important difference in
Algorithm~\ref{alg:interleaved-level-search} compared with
Algorithm~\ref{alg:simple-level-search} is that it \emph{defers}
inserting tree edges found on this level until the end of the search.
Instead, it continues to search for replacement edges from the
\emph{initial} components until the component is deactivated. This
property is important to show that the work done for a component
across all rounds is dominated by the cost of the last round, since the number of
vertices in the component is fixed, but the number of non-tree edges
examined doubles in each round. For the same reason, it also defers
inserting the pushed edges onto level $i - 1$. We crucially use this
property to obtain improved batch work bounds (Section~\ref{sec:batch}).

% Push everything
Another difference in the modified algorithm is that if a
component is still active after adding the replacement edges found in
this round (i.e., the component on level $i$ still has size at most $2^{i-1}$), then \emph{all} of the edges found in this round can be pushed to level $i - 1$ without violating Invariant~\ref{inv:component_sizes}. Notice now that
when pushing down edges, both the \emph{tree and non-tree} edges
that are found in this round are pushed.
Pushing down all edges ensures that the algorithm performs enough level decreases to
which to charge the work performed during the next round.
The component deactivates either once it runs out of incident non-tree
edges, or when it becomes too large. Since the algorithm defers adding
the new tree edges found until the end of the level, it also maintains
an auxiliary data structure that dynamically tracks the size of the
resulting components as new edges are found.

% % Round complexity
% In this section, we show how to avoid starting over by maintaining a
% \emph{single, geometrically increasing search size} over all rounds of
% the search. In the new algorithm, each component that is small enough
% to search at this level starts out as active, and is deactivated
% exactly once by the algorithm, either when it runs out of non-tree
% edges, or becomes too large to search due to tree-edges that we
% discover. We show that each active component that is active up until
% some round $R > 1$ pushes down $2^{r}$ edges on every round $r < R$.
% The algorithm, therefore, terminates in $O(\log^2 n)$ rounds, and by
% implementing each round in $O(\log n)$ depth w.h.p. we obtain a
% $O(\log^3 n)$ depth parallel algorithm over all levels.

\myparagraph{The Deletion Algorithm}
We briefly describe the main differences between \intlevelsearch{},
the new level search procedure, and \parlevelsearch{}. The algorithm
consists of a number of \defn{rounds}
(Lines~\ref{ils:loopstart}--\ref{ils:loopend}). We use $r$ to track
the round numbers, and we use $E_P$ to store the set of both tree and
non-tree edges that will be pushed to level $i-1$ at the end of the
search at this level (Line~\ref{ils:opdefs}).  $T$ stores the set of
tree edges that have been selected, which will be added to the
spanning forest at the end of the level. Lastly, we use $M$ to
maintain a dynamic mapping from all the components in $L$ to a unique
representative for their contracted supercomponent (initially itself),
and the size of the contracted supercomponent.

In round $r$, the algorithm first retrieves the first $2^{r}$ (or
fewer) edges of each the active components in parallel, and finds
replacement edges. All replacement edges are added to the set $R$
(line~\ref{ils:findreplacements}).

The algorithm then computes a spanning forest over the edges in
$R$, and computes $T_{r}$, which are the original replacement edges in
$R$ that were selected as spanning forest edges
(lines~\ref{ils:sfcomp}--\ref{ils:addtoT}). The spanning forest
computation returns, in addition to the tree edges, a mapping from the
vertices in $R'$ to their connectivity label (line~\ref{ils:sfcomp}),
which can be used on line~\ref{ils:updateM} to efficiently update the
representatives of all affected components and the sizes of the
supercomponents.

The next step maps over the components in parallel again, calling
\textsc{PushEdges} on each active component, and checks whether the
edges searched in this round can be (lazily) pushed to
level $i-1$ (Line~\ref{ils:findedgestopush}).\footnote{
  Note that the set of edges retrieved by \textsc{PushEdges} in
Line~\ref{fetp:fetchedges} is assumed to be the same as the one in
Line~\ref{fr:fetchedges}. This assumption is satisfied by using our
\textsc{FetchEdges} primitive on a \batchparallel{} ET-tree, and can
be satisfied in general by associating the edges retrieved in
\textsc{ComponentSearch} to be used in \textsc{PushEdges}.
}
If a component is still active (its new size is small enough to still
be searched, and the component still has some non-tree edges
remaining) (line~\ref{fetp:pushingcondition}), all of the searched
edges are removed from the adjacency lists at level $i$
(line~\ref{fetp:deletefromleveli}) and are added to the set of edges
that will be pushed to level $i-1$ at the end of the level
(Lines~\ref{fetp:returnec} and~\ref{ils:findedgestopush}). Note that
this set of edges contains both replacement tree edges we discovered,
and non-tree edges. The tree-edges can be pushed down to level $i-1$
because the component with the tree edges added has size $\leq
2^{i-1}$.

%If no replacement edges were found in this round, and the search size
%was equal to the number of non-tree edges, or if the new component
%size ($c_{sz}'$) is small enough for the component to remain active in
%the next round, then we push all searched edges to level $i-1$
%(lines~\ref{ils:pushlineone}--\ref{ils:pushlinetwo}).

The end of the round
(lines~\ref{sls:setQ}--\ref{ils:rupdate}) handles updating
the set of components and incrementing the round number, as in
Algorithm~\ref{alg:simple-level-search}.

%We first compute $D_r$, the set of components in
%$C$ that are no longer active, i.e., either become too large, or have
%no non-tree edges left to examine (line~\ref{ils:nolongeractive}). Next,
%we add these components to $D$ (line~\ref{ils:dupdate}) and remove them
%from $C$ so that they are not searched on the next round
%(line~\ref{ils:rupdate}).

Finally, once all components are inactive, the tree edges found at
this level that are not contained in $E_p$ are promoted (the tree
edges added to $E_p$ have their level decreased to $i-1$) and inserted
into $F_{i}$ (Lines~\ref{ils:promotetree}--\ref{ils:batchinserttree}),
and all edges added to $E_P$ in Line~\ref{ils:findedgestopush} are
pushed down to level $i-1$ (Line~\ref{ils:pushdown}). Note that any
tree-edges found in this set are promoted in level $i-1$ and added to
$F_{i-1}$. The procedure returns the set of components and all
replacement edges found at this level and levels below it
(Line~\ref{ils:return}).

\subsection{Cost Bounds}
We start by showing that the depth of
Algorithm~\ref{alg:interleaved-level-search} is $O(\log^3 n)$.

\begin{lemma}\label{lem:interleave-rounds}
The number of rounds performed by
Algorithm~\ref{alg:interleaved-level-search} is $O(\log n)$ and the
depth of each round is $O(\log n)$ w.h.p.. The depth of the
\intlevelsearch{} is therefore $O(\log^2 n)$ w.h.p..
\end{lemma}
\begin{proof}
Each round of the algorithm increases the search size of a component
by a factor of $2$. Therefore, after $O(\log n)$ rounds, every
non-tree edge incident on a component will be considered and the
algorithm will terminate.

To argue the depth bound, we consider the main steps performed during
a round. Fetching, examining and removing the edges from level $i$ takes $O(\log n)$
depth w.h.p.\ by Lemma~\ref{lem:ettree-fetch-edges}, Theorem~\ref{thm:ett-bounds},
and Lemma~\ref{lem:ettree-decrease-level}. Computing a spanning forest on the
replacement edges and filtering the components (at most $k$
replacement edges, or components) can be done in $O(\log k)$ depth.
The depth per round is therefore $O(\log n)$ w.h.p. and the depth of
\intlevelsearch{} is $O(\log^2 n)$ w.h.p.
\end{proof}

Combining Lemma~\ref{lem:interleave-rounds} with the fact that there
are $\log n$ levels gives the following theorem.

\begin{theorem}\label{lem:interleaved-depth}
A batch of $k$ edge deletions can be processed in $O(\log^3 n)$ depth w.h.p.
\end{theorem}

We now consider the work performed by the algorithm.  We start with a
lemma showing that the search-size for a component increases
geometrically until the round where the component is deactivated.

\begin{lemma}\label{lem:component-doubling}
Consider a component, $c$, that is active at the end of round $r-1$.
If $c$ is not removed from $C$, then it examines $\geq 2^{r-1}$ edges
that are pushed down to level $i-1$ at the end of the search.
\end{lemma}
\begin{proof}

We prove the contrapositive. Suppose that $< 2^{r-1}$ edges are pushed
down in total by $c$ in the last round.  Then, we will show that $c$
cannot be active in the next round (i.e., it is removed from $C$ in
round $r-1$).

Notice that $c$ must be active at the start of round $r-1$.  Consider
the check on Line~\ref{fetp:pushingcondition}, which checks whether
$w \leq 2^{r-1}$ and $w < w_{\max}$ on this round. Suppose for the
same of contradiction that both conditions are true. Then, by
the fact that $w <  w_{\max}$, it must be the case that $w = 2^{r-1}$
by Line~\ref{fetp:wdef}. If the condition is true, then on
Line~\ref{fetp:deletefromleveli} the algorithm adds $2^{r-1}$ edges to
be pushed to level $i-1$, contradicting our assumption that $< 2^{r-1}$ edges are pushed.

Therefore the check on Line~\ref{fetp:pushingcondition} must be false,
giving that either $w > 2^{i-1}$, or $w =  w_{\max}$. This
means that $c$ will be marked as inactive on Line~\ref{sls:setQ},
and then become deactivated on line~\ref{ils:cupdate}.
Therefore, if $< 2^{r-1}$ edges are pushed down by $c$ in round $r-1$,
$c$ is deactivated at the end of the round, concluding the proof.

%Let $c_1, \ldots, c_d$ be the components that joined $c$ on this
%round. Observe that the search size on round $r-1$ is $2^{r-1}$ (line
%8). Therefore, if $< 2^{r-1}$ edges are searched in total, then
%$c_{sz} < 2^{r-1}$ for each $c_i$. That is, the number of non-tree
%edges incident to $c_i$ is less than $2^{r-1}$, and so $c_i$ examines
%all of the non-tree edges incident to it.
%
%We now consider whether each $c_i$ pushes all of these non-tree edges
%down. The crucial line is the check on line 25. If either condition is
%true, then we are fine since the non-tree edges are pushed down.
%However, if the check fails, then we have that $c'_{sz} > 2^{i-1}$ and
%so $c$ is too large to search in the next round. Therefore, $c$ will
%be included in the set of elements that are removed from $C$ in line
%31. Thus, we have that either the check on line 25 passes for all
%$c_i, 1 \leq i \leq d$, or $c$ becomes too large and is deactivated on
%this round.  Finally, if the check passes for all $c_i$, then all
%non-tree edges are pushed down for each $c_i$ (lines 26--27). Since
%$c_{sz} < 2^{r-1}$ for each $c_i$, $c_{sz} =  c_{\max}$, and therefore
%no non-tree edges remain incident to $c_i$. Therefore, $c$ has no
%non-tree edges incident to it, and it will be deactivated in this
%round.
\end{proof}

\begin{lemma}\label{lem:component-root-dominated}
	Consider the work done by some component $c$ over the course of
  \intlevelsearch{} at a given level. Let $R$ be the total number
  rounds that $c$ is active. Then, $c$ pushes down $p_{c} = 2^{R} - 1$ edges
  in total. Furthermore, the total cost of searching for and pushing
  down replacement edges performed by $c$ is
  \begin{equation}
    O\left(p_{c}\log\left(1 + \frac{n_{c}}{p_{c}}\right)\right)
  \end{equation}
  in expectation, where $n_{c}$ is the number of vertices in $c$.
\end{lemma}
\begin{proof}
By Lemma~\ref{lem:component-doubling}, for each round $r < R$, $c$
adds $2^{r}$ edges to be pushed down. Summing over all
rounds shows that the total number of edges added to be pushed down is
$2^{R} - 1$. The cost of pushing down these edges at the end of the search
at this level is exactly
\begin{equation}
    O\left(p_{c}\log\left(1 + \frac{n_{c}}{p_{c}}\right)\right).
\end{equation}
by Lemma~\ref{lem:ettree-decrease-level}, since the size of the tree
that is affected is $n_c$.

We now consider the cost of fetching and examining the edges over all
rounds.
The cost of fetching and examining $2^{r}$ edges is
\begin{equation}
  O\left(2^{r}\log\left(1 + \frac{n_{c}}{2^{r}}\right)\right),
\end{equation}
in expectation by Theorem~\ref{thm:ett-bounds} and Lemma~\ref{lem:ettree-fetch-edges}.
Summing over all rounds $r < R$, the work is
\begin{equation}
  \sum_{r=1}^{R-1} O\left(2^{r}\log\left(1 + \frac{n_{c}}{2^{r}}\right)\right)
\end{equation}
in expectation to fetch and examine edges in the first $R-1$ rounds, which is equal to
\begin{equation}
  O\left(2^{R}\log\left(1 + \frac{n_{c}}{2^{R}}\right)\right),
\end{equation}
by Lemma~\ref{lem:batch_bound_is_root_dominated}. Since on round $R$,
the algorithm searches at most $2^{R}$ edges, the total cost of searching
for replacement edges over all rounds is at most
\begin{equation}
  O\left(2^{R}\log\left(1 + \frac{n_{c}}{2^{R}}\right)\right) =
    O\left(p_{c}\log\left(1 + \frac{n_{c}}{p_{c}}\right)\right).
\end{equation}

\end{proof}

\begin{lemma}\label{lem:interleaved_level_search_cost}
	The cost of \textsc{InterleavedLevelSearch} is at most
	\begin{equation}
	O\left(k \log\left(1 + \frac{n}{k} \right) +
         p \log\left(1 + \frac{n}{p} \right)\right)
	\end{equation}
  in expectation where $p$ is the total number of edges pushed down.
\end{lemma}

\begin{proof}
  First consider lines 2--5. Since we are deleting a batch of $k$
  edges, we can find at most $k$ replacement edges to reconnect these
  components. Therefore line 2 performs $O\left(k \log\left(1 +
      \frac{n}{k} \right)\right)$ expected work by Theorem~\ref{thm:ett-bounds}.
  Pushing $t$ spanning tree edges to the next level (line 5) can be
  done in $O\left(t \log\left(\frac{n}{t} + 1 \right)\right))$
  expected work by
  Lemmas~\ref{lem:ettree-fetch-edges},~\ref{lem:ettree-decrease-level},
  and~\ref{lem:component_bounds}, and Theorem~\ref{thm:ett-bounds}.
  Hence in total, lines 2--5 perform at most $O\left(k \log\left(1 +
      \frac{n}{k}\right) + t \log\left(1 + \frac{n}{t}\right) \right)$
  work in expectation.

  Now, consider the cost of the steps which scan or update the
  components that are active in each round. On the first round, this
  cost is $O(k)$. In every subsequent round, $r$, by
  Lemma~\ref{lem:component-doubling} each currently active component must
  have added $2^{r-1}$ edges to be pushed down on the previous
  round. Therefore, we can charge the $O(1)$ work per
  component performed in this round to these edge pushes.

  Next, we analyze the work done while searching for and pushing replacement edges.
  Consider some component $c \in C$ that is searched on this level. By
  Lemma~\ref{lem:component-root-dominated}, the cost of searching for
  and pushing down the replacement edges incident on this component
  is
  \begin{equation}
    O\left(p_{c}\log\left(1 + \frac{n_{c}}{p_{c}}\right)\right)
  \end{equation}
  in expectation, where $n_{c}$ is the number of vertices in $c$ and
  $p_{c}$ is the total number of edges pushed down by $c$.

  The total work done over all components to search for replacement edges
  and push down both the original tree edges, and the edges
  in each round is therefore
  \begin{equation}
    O\left( t \log\left(1 + \frac{n}{t}\right) +
    \sum_{c \in C} p_{c}\log\left(1 + \frac{n_{c}}{p_{c}}\right) \right).
  \end{equation}
  in expectation. Since $\sum n_{c} = n$, by Lemma~\ref{lem:component_bounds} this
  costs
  \begin{equation}
    O\left(p\log\left(1 + \frac{2n}{p}\right)\right) = O\left(p\log\left(1 + \frac{n}{p}\right)\right)
  \end{equation}
  work in expectation, where $p =  t + \sum p_{c}$ is the total number
  of edges pushed, including tree and non-tree edges.
  Therefore, the total cost is
	\begin{equation}
	O\left(k \log\left(1 + \frac{n}{k} \right) +
         p \log\left(1 + \frac{n}{p} \right)\right)
	\end{equation}
  in expectation.
\end{proof}

\begin{theorem}\label{thm:interleaved-work}
The expected amortized cost per edge insertion or deletion is $O(\log^2 n)$.
\end{theorem}
\begin{proof}
  The proof follows from the same argument as
  Theorem~\ref{thm:simple-work}, by using
  Lemma~\ref{lem:interleaved_level_search_cost}.
\end{proof}

\subsection{A Better Work Bound}\label{sec:batch}

We now show that by a more careful analysis, we can obtain a tighter
bound on the amount of work performed by the interleaved algorithm.
In particular, we show in this section that the algorithm performs
\begin{equation}
O\left(\log n \log\left( 1 + \frac{n}{\bark{}} \right) \right)
\end{equation}
amortized work per edge in expectation, where $\bark{}$ is the average
batch size of all batches of deletions.
Therefore, if we process batches of deletions of size
$\Omega(n/\textnormal{polylog}(n))$ on average, our algorithm performs $O(\log n
\log\log n)$ expected amortized work per edge, rather than $O(\log^2
n)$. Furthermore, if we have batches of size $\Omega(n)$, the cost is
just $O(\log n)$ per edge.

At a high level, our proof formalizes the intuition that in the worst
case, all edges are pushed down at every level, and that performing
fewer deletion operations results in larger batches of pushes which
take advantage of work bounds of the ET-tree.
Our proof crucially relies on the fact that although the deletion
algorithm at a level can perform $O(\log n)$ ET-tree operations per
component, since the batch sizes are geometrically increasing, these
operations have the cost of a single ET-tree operation per component.
Furthermore, Lemma~\ref{lem:interleaved_level_search_cost} shows that
the costs per component can be combined so that the total cost is
equivalent to the cost of a single ET-tree operation on all the
vertices. Therefore, the number of deletion operations can be exactly
related to the effective number of ET-tree operations at a level.
We relate the number of deletions to the average batch size, which
lets us obtain a single unified bound for both insertions and
deletions.

% (for which the
% average batch deletion size is a proxy)

%formalizes the
%intuition that performing fewer batches of deletions
%
%Consider the lifetime of an edge. In the worst case, the edge
%eventually pushed down at every level.

%In the worst case, we delete all edges eventually.
%Due to the batch bounds on the ET-trees, doing a few big batches is
%better than many small ones.
%
%Average deletion batch size is just a proxy for the number of deletion
%operations that we do.

%\guy{aren't we missing a loglog term.
%  Also it would be nice if we gave an intuition somewhere, perhaps
%  earlier, why the bound improves.}\laxman{fixed: daniel, can you add
%  the intuition? we discussed this earlier, I think it would be good
%  to add either here, or in the intro.}

\begin{theorem}
	Using the interleaved deletion algorithm, the amortized work performed by \textsc{BatchDeletion} and \textsc{BatchInsertion} on a batch of $k$ edges is
	\begin{equation}
	O\left(k\log n \log\left( 1 + \frac{n}{\bark{}} \right) \right),
	\end{equation}
  	in expectation where $\bark{}$ is the average batch size of all batch deletions.
\end{theorem}

\newcommand{\sumwork}{\sum_{\textnormal{\tiny batch } b} \sum_{\textnormal{\tiny level } i}}

\begin{proof}
  Batch insertions perform only $O\left(k \log\left(1 +
      \frac{n}{k}\right)\right)$ work by Theorem~\ref{thm:insertions},
  so we focus on the cost of deletion since it dominates. Consider the total amount of
  work performed by all batch deletion operations at any given point in
  the lifetime of the data structure. We will denote by $k_b$, the
  size of batch $b$, and by $p_{b,i}$, the number of edges pushed down
  on level $i$ during batch $b$.  Combining Lemmas~\ref{lem:batch_delete_exclude_search},
  and \ref{lem:interleaved_level_search_cost}, the total work is bounded
  above by
  \begin{equation}\label{eqn:total_work}
  \begin{split}
  O\left( \sumwork k_b \log\left(1 + \frac{n}{k_b}\right) + p_{b,i} \log\left(1 + \frac{n}{p_{b,i}} \right) \right).
  \end{split}
  \end{equation}
  We begin by analyzing the first term, which is paid for by the deletion algorithm. Let
  \begin{equation}
  K = \sum_{\textnormal{\tiny batch } b} k_b
  \end{equation}
  denote the total number of deleted edges. Applying Lemma~\ref{lem:component_bounds}, and using the fact that there are $\log n$ levels, we have
  \begin{equation}
  O\left( \sumwork k_b \log\left(1 + \frac{n}{k_b}\right)\right) = O\left( K \log n\log\left(1 + \frac{n \cdot d}{K} \right)\right),
  \end{equation}
  where $d$ is the number of batches of deletions. Since $K / d = \bark{}$, this
  is equal to
  \begin{equation}
  O\left( K \log n \log\left(1 + \frac{n}{\bark{}} \right)\right),
  \end{equation}
  work in expectation. Each batch can therefore be charged a cost of $\log n\log\left(1 + n/\bark{} \right)$ per edge, and hence the amortized cost of batch deletion is
  \begin{equation}
  O\left(k\log n \log\left( 1 + \frac{n}{\bark{}} \right) \right)
  \end{equation}
  in expectation.

  The remainder of the cost, which comes entirely from searching for
  replacement edges, is charged to the insertions. Consider this cost and let
  \begin{equation}
  P = \sumwork p_{b,i}
  \end{equation}
  denote the total such number of edge pushes. Since the total number of terms
  in the double sum is $d\log n$, Lemma~\ref{lem:component_bounds} allows\
  us to bound the total work of all pushes by
  \begin{equation}
  \sumwork p_{b,i} \log\left(1 + \frac{n}{p_{b,i}} \right) = O\left( P \log\left( 1 + \frac{ n d \log n}{P} \right)\right).
  \end{equation}
  in expectation. Since every edge can only be pushed down once per level, we have
  \begin{equation}
  P \leq m \log n,
  \end{equation}
  where $m$ is the total number of edges ever inserted. Therefore by Lemma~\ref{lem:batch_bound_is_increasing}, the total work is at most
  \begin{equation}
  O\left( m \log n \log\left(1 + \frac{n d \log n}{m \log n} \right) \right) = O\left( m \log n \log\left(1 + \frac{n d}{m} \right) \right)
  \end{equation}
  in expectation. Since $d = K / \bark{}$, this is equal to
  \begin{equation}
  O\left( m \log n \log\left(1 + \frac{nK}{m\bark{}} \right) \right)
  \end{equation}
  in expectation. Since each edge can be deleted only once, we have $K \leq m$, and hence we obtain that the total work to push all tree edges down is at most
  \begin{equation}
  \begin{split}
  O\left( m \log n \log\left(1 + \frac{n}{\bark{}} \right) \right).
  \end{split}
  \end{equation}
  in expectation. We can therefore charge $ O\left( \log n \log(1 + n/\bark{}) \right)$ per edge to each batch insertion. Since this dominates the cost of the insertion algorithm itself, the amortized cost of batch insertion is therefore
  \begin{equation}
  O\left(k\log n \log\left( 1 + \frac{n}{\bark{}} \right) \right),
  \end{equation}
  in expectation as desired, concluding the proof.
\end{proof}

\section{Related Work}
\myparagraph{Parallel Dynamic Algorithms}
There are only a few results on parallel dynamic algorithms.
Earlier results~\cite{ferragina1994batch, das1994n}  are not work-efficient with respect to the fastest sequential dynamic algorithms, do not support batch updates, and perform polynomial work per update. 
Some more recent results such as parallel dynamic depth-first
search~\cite{khan17dfs} and minimum spanning
forest~\cite{kopelowitz18mst} process updates one at a time, and are therefore not \batchdynamic{} algorithms.
Work efficient parallel \batchdynamic{} algorithms include those for  the well-spaced point sets problem~\cite{acar2011parallelism} and those for the dynamic trees problem~\cite{reif94treecontraction,acar2017brief,tseng2018batch}.

\myparagraph{Parallel Connectivity}
Parallel algorithms for connectivity have a long
history~\cite{Hirschberg1979, ShiloachV82, Vishkin1984, AwerbuchS83,
Reif85, Phillips89, Cole1991, Karger1999}, and there are many existing
algorithms that solve the problem work-efficiently and in
low-depth~\cite{gazit1991optimal,ColeKT96, HalperinZ94, Halperin00,
PettieR02, PoonR97, Shun2014}, some of which are also
practical~\cite{Shun2014, dhulipala2018theoretically}.  However, there
is no obvious way to adapt existing parallel connectivity algorithms
to the dynamic setting, particularly for batch updates.

\myparagraph{Parallel Dictionaries and Trees}
There are many results on parallel dictionaries and trees supporting
batch
updates~\cite{gil1991towards,Blelloch1998, blelloch1999pipelining,
shun2014phase, BlellochFS16, akhremtsev2016fast, sun2018parallel}.
The dictionary data structures in the literature culminated in
dictionaries supporting batch insertions, deletions and lookups in
linear work and $O(\log^{*} n)$ depth w.h.p.~\cite{gil1991towards}.
Early work on batch insertions into trees focused on optimizing the
depth, but was not work-efficient. Paul et al. design batch search,
insertion and deletion algorithms for 2-3 trees on the EREW
PRAM~\cite{paul1983parallel}. These results were later extended to
B-trees by Higham et al.~\cite{higham1994maintaining}. The algorithms
of both Paul et al.  and Higham et al. perform $O(m \log n)$ work for
$m$ tree operations.

Recent work on parallel tree data structures has focused on how to
parallelize batch operations for various balancing schemes in binary
search trees~\cite{BlellochFS16}, and also how to improve the depth of
these operations~\cite{akhremtsev2016fast}. There is also some very
recent work on extending these tree data structures to support range
and segment queries~\cite{sun2018parallel} as well as practical
implementations of parallel trees supporting batch insertions,
deletions and lookups~\cite{pam}.

\myparagraph{Other Related Work}
There is also recent work on parallel working-set structures that
supports batching by Agrawal et al.~\cite{agrawal2018working}.
Earlier work by Agrawal et al.~\cite{agrawal2014batching} introduces
the idea of implicit batching which uses scheduler support to convert
dynamically multithreaded programs using an abstract data type to
programs that perform batch accesses to an underlying parallel data
structure.

\section{Discussion}
%\umut{Overall, Too many ``the first'' (loses impact).}
%\umut{``the first algorithms'', sound off. can't be more than one the first.}

In this paper, we present a novel \batchdynamic{} algorithm for the
connectivity problem. Our algorithm is always work-efficient with
respect to the Holm, de Lichtenberg and Thorup dynamic connectivity
algorithm, and is asymptotically faster than their algorithm when the
average batch size is sufficiently large.
A parallel implementation of our algorithm achieves $O(\log^3 n)$
depth w.h.p., and is, to the best of our knowledge, the first parallel
algorithm for the dynamic connectivity problem performing $O(T
\mathop{\text{polylog}}(n))$ total expected work, where $T$ is the total number
of edge operations.

There are several natural questions to address in future work. First,
can the depth of our algorithm be improved to $O(\log^2 n)$ without
increasing the work? Investigating lower bounds in the
batch setting would also be very interesting---are there non-trivial
lower-bounds for \batchdynamic{} connectivity? Lastly, in this paper
we show expected amortized bounds. One approach to strengthen these
bounds is to show that our tree operations hold w.h.p.\ and argue
that our amortized bounds hold w.h.p. Another is to design a
deterministic \batchdynamic{} forest connectivity data structure with
the same asymptotic complexity as the \batchparallel{} ET-tree, which
would make the randomized bounds in this paper deterministic.
%\daniel{except we use a dictionary for our adjacency lists}

Two additional questions are whether we can extend our results to give
parallel work-efficient \batchdynamic{} MST, 2-edge connectivity and
biconnectivity algorithms. MST seems solvable using the techniques
presented in this paper, although our dynamic tree structure would
need to be extended with additional primitives. Existing sequential
2-edge connectivity and biconnectivity algorithms require a dynamic
tree data structure supporting path queries which are not supported by
ET-trees. However, RC-trees~\cite{acar2017brief} can be extended to
support path queries, which makes them a possible candidate for this
line of work.
Finally, it seems likely that ideas from our work can be
extended to give a parallel \batchdynamic{} Monte-Carlo connectivity
algorithm based on the Kapron-King-Mountjoy
algorithm~\cite{kapron2013dynamic}.

\section*{Acknowledgments}

We thank Tom Tseng and Goran Zuzic for helpful discussions. This work was supported in part by NSF grants CCF-1408940, CCF-1533858, and CCF-1629444.

\bibliographystyle{abbrv}
\bibliography{ref}

\appendix
\section{Model}\label{sec:app-model}

The \emph{Multi-Threaded Random-Access Machine (\mpram{})} consists of a set of
\processes{} that share an unbounded memory.  Each \process{} is
essentially a Random Access Machine---it runs a program stored
in memory, has a constant number of registers, and uses standard RAM
instructions (including an \insend{} to finish the computation).  The
only difference between the \mpram{} and a RAM is the \forkins{} instruction that takes a positive
integer $k$ and forks $k$ new child \processes{}.  Each child \process{}
receives a unique identifier in the range $[1,\ldots,k]$ in its first
register and otherwise has the same state as the parent, which
has a $0$ in that register. All children start by running the next
instruction.  When a \process{} performs a \forkins{}, it is suspended until
all the children terminate (execute an \insend{} instruction).  A
computation starts with a single root \process{} and finishes when that
root \process{} terminates. This model supports what is often referred to as
nested parallelism. Note that if root \process{} never forks, it is a
standard sequential program.

We note that we can simulate an \mpram{} algorithm on the CRCW PRAM
equipped with the same operations with an additional $O(\log^{*} n)$
factor in the depth due to load-balancing. Furthermore, a PRAM
algorithm using $P$ processors and $T$ time can be simulated in our
model with $PT$ work and $T$ depth. We equip the model with a
compare-and-swap operation (see Section~\ref{sec:prelims}) in this
paper.

Lastly, we define the cost-bounds for this model.  A computation can
be viewed as a series-parallel DAG in which each instruction is a
vertex, sequential instructions are composed in series, and the forked
sub\processes{} are composed in parallel. The \defn{work} of a
computation is the number of vertices and the \defn{depth}
(\defn{span}) is the length of the longest path in the DAG. We refer
the interested reader to~\cite{blelloch18notes} for more details about
the model.

\section{Data Structures}\label{sec:datastructures}

In this section we describe a simple adjacency-list like data
structure that efficiently supports insertion and deletion of arbitrary edges,
and quickly fetching a batch of $l$ edges. This is the data structure
that we use to store adjacency lists of vertices at each level.
Note that we actually store two adjacency lists, one for tree edges,
and one for non-tree edges. The adjacency list data structure supports
the following operations:

\vspace{0.5em}
\begin{itemize}[topsep=0pt,itemsep=0ex,partopsep=0ex,parsep=1ex, leftmargin=*]
  \item \textbf{$\textproc{InsertEdges}(\{e_1, \ldots, e_l\})$}:
    Insert a batch of edges adjacent to this vertex.

  \item \textbf{$\textproc{DeleteEdges}(\{e_1, \ldots, e_l\})$}:
    Delete a batch of edges adjacent to this vertex.

  \item \textbf{$\textproc{FetchEdges}(l)$}: Return a set of $l$
    arbitrary edges adjacent to this vertex.
\end{itemize}
\vspace{0.5em}

We now show how to implement a data structure that gives us the
following bounds:

\adjacencyops*

\begin{proof}
For a given vertex, the data structure stores a list of pointers to
each adjacent edge in a resizable array. Each edge correspondingly
stores its positions in the adjacency arrays of its two endpoints.
Since each vertex can have at most $O(n)$ edges adjacent to it, the
adjacency arrays are of size at most $O(n)$.

Insertions are easily handled by inserting the batch onto the end of the
array, and resizing if necessary. This costs $O(1)$ amortized work
per edge and $O(\log n)$ depth. To fetch $l$ elements, we simply return
the first $l$ elements of the array, which takes $O(1)$ work per edge and
$O(\log n)$ depth.

Finally, to delete a batch of $l$ edges, the algorithm first determines which
of the edges to be deleted are contained within the final $l$ elements of
the array. It then compacts the final $l$ elements of the array, removing
those edges. Compaction costs $O(l)$ work and $O(\log n)$ depth. The
algorithm then considers the remaining $l'$ edges to be deleted, and
in parallel, swaps these elements with the final $l'$ elements of the array.
The final $l'$ elements in the array can then be safely removed. Note that
any operation that moves an element in the array also updates the
corresponding position value stored in the edge. Swapping and deleting
can be implemented in $O(l')$ work and $(\log n)$ depth, and hence all
operations cost $O(1)$ amortized work per edge and $O(\log n)$ depth.
\end{proof}

\section{Additional Tree Operations}\label{sec:extra-tree-ops}

\myparagraph{Retrieving and Pushing Down Edges}
The \batchparallel{} ET-trees used in this paper augment each node in
the tree with two values indicating the number of tree and non-tree
edges whose level is equal to the level of the forest currently stored
in that subtree. The augmentation is necessary for efficiently
fetching the tree edges that need to be pushed down before searching
the data structure, and for fetching a subset of non-tree edges in a
tree.

We extend the \batchdynamic{} trees interface described earlier with
operations which enable efficiently retrieving, removing and pushing
down batches of tree or non-tree edges.

These primitives are all similar and can be implemented as follows. We
first describe the primitives which fetch and remove a set of $l$ tree
(or non-tree) edges. The algorithm starts by finding a set of vertices
containing $l$ edges.
To do this we perform a binary search on the skip-list in order to
find the first node that has augmented value greater than $l$. The
idea is to sequentially walk at the highest level, summing the
augmented values of nodes we encounter and marking them, until the
first node that we hit whose augmented value makes the counter larger
than $l$, or we return to $v$. In the former case, we descend a level
using this node's downwards pointer, and repeat, until we reach a
level 0 node. We also keep a counter, $ctr$, indicating the number of
tree (non-tree) edges to take from the rightmost marked node at level
0. Otherwise, all nodes at the topmost level are marked. The last step
of the algorithm is to find all descendants of marked nodes that have
a non-zero number of tree (non-tree) edges, and return all tree
(non-tree) edges incident on them. The only exception is the rightmost
marked node, from which we only take $ctr$ many tree (non-tree) edges

Insertions are handled by first inserting the edges into the adjacency
list data structure. We then update the augmented values in the
ET-tree using the primitive from Tseng et al.~\cite{tseng2018batch}.

We now argue that these implementations achieves good work and depth
bounds.

\ettreefetchedges*

\begin{proof}
Standard proofs about skip-lists shows that the number of nodes
traversed in the binary search is $O(\log n)$ w.h.p.~\cite{pugh90,
  tseng2018batch}. We can fetch $l$ edges from each
vertex's adjacency list data structure in $O(l)$ amortized work and
$O(\log n)$ depth by Lemma~\ref{lem:adjacencyops}. The total work is
therefore $O\left(l \log\left(1 + \frac{n_c}{l}\right)\right)$ in
expectation, and the depth is $O(\log n)$ w.h.p.\ since the depth of
the adjacency list access is an additive increase of $O(\log n)$.
Observe that removing the edges can be done in the same bounds
since updating the augmented values after deleting the edges costs
$O\left(l \log\left(1 + \frac{n_c}{l}\right)\right)$ expected work.
\end{proof}

\ettreedecreaselevel*

\begin{proof}
  The proof is identical to the proof of
  Lemma~\ref{lem:ettree-fetch-edges}. The only difference is that the
  augmented values of the nodes that receive an edge must be updated
  after insertion which costs at most $O\left(l \log\left(1 +
      \frac{n_c}{l}\right)\right)$ in expectation. Note that since the
  forest on the lower level is a subgraph of the tree at the current
  level, it has size at most $n_c$, proving the bounds.
\end{proof}

\section{Additional Proofs}\label{sec:additional-proofs}

We now state and give proofs for some of the technical lemmas used in
our proofs of the improved batch bounds.

\componentbounds*
\begin{proof}
	We proceed by induction on $c$. When $c = 1$, the quantities are equal. For $c > 1$, we can write
	\begin{equation}
	\begin{split}
		\sum_{i=1}^c k_i \log\left(1 + \frac{n_i}{k_i}\right) &= \sum_{i=1}^{c-1} k_i \log\left(1 + \frac{n_i}{k_i}\right) + k_c \log\left(1 + \frac{n_c}{k_c}\right), \\
		&\leq (k - k_c) \log\left(1 + \frac{n - n_c}{k - k_c}\right) + k_c \log\left(1 + \frac{n_c}{k_c}\right).
	\end{split}
	\end{equation}
	Then, using the concavity of the logarithm function, we have
	\begin{equation}
		\begin{split}
		\sum_{i=1}^c k_i \log\left(1 + \frac{n_i}{k_i}\right) &\leq k \log\left( \frac{k - k_c}{k} \left(1 + \frac{n - n_c}{k - k_c}\right) + \frac{k_c}{k} \left( 1 + \frac{n_c}{k_c} \right) \right), \\
		&= k \log\left( \frac{k - k_c}{k} + \frac{n - n_c}{k} + \frac{k_c}{k} + \frac{n_c}{k} \right), \\
		&= k \log\left(1 + \frac{n}{k} \right),
		\end{split}
	\end{equation}
	which concludes the proof.
\end{proof}

\batchboundisrootdominated*
\begin{proof}
	First, write
	\begin{equation}
	\begin{split}
	\log\left(1 + \frac{n}{2^w}\right) &= \log\left(1 + 2^{r-w} \frac{n}{2^r}\right), \\
	&\leq \log \left(2^{r-w}\left(1 + \frac{n}{2^r} \right)\right), \\
	&= \log(2^{r-w}) + \log\left( 1 + \frac{n}{2^r} \right), \\
	&= (r-w) + \log\left( 1 + \frac{n}{2^r} \right).
	\end{split}
	\end{equation}
	Now substitute this into the sum to obtain
	\begin{equation}
	\begin{split}
	\sum_{w=0}^r 2^w \log\left(1 + \frac{n}{2^w}\right) \leq \sum_{w=0}^r (r-w)2^w  + \log\left( 1 + \frac{n}{2^r} \right) \sum_{w=0}^r 2^w, \\
	= \sum_{w=0}^r (r-w)2^w + O \left( 2^r \log\left( 1 + \frac{n}{2^r} \right) \right),
	\end{split}
	\end{equation}
	We evaluate the remaining sum by writing
	\begin{equation}
	\sum_{w=0}^r (r-w)2^w = \sum_{w=0}^r \frac{r - w}{2^{r-w}} 2^r,
	\end{equation}
	and then use the fact that
	\begin{equation}
	\sum_{w=0}^r \frac{r - w}{2^{r-w}} = O(1)
	\end{equation}
	to conclude that
	\begin{equation}
	\begin{split}
	\sum_{w=0}^r 2^w \log\left(1 + \frac{n}{2^w}\right) &= O(2^r) + O \left( 2^r \log\left( 1 + \frac{n}{2^r} \right) \right), \\
	&= O \left( 2^r \log\left( 1 + \frac{n}{2^r} \right) \right),
	\end{split}
	\end{equation}
	as desired.
\end{proof}

\batchboundisincreasing*
\begin{proof}
	The derivative of the function with respect to $x$ is
	\begin{equation}
	\log\left(1 + \frac{n}{x}\right) - \frac{n}{n+x}.
	\end{equation}
	We must show that this quantity is strictly positive for all
	$x \geq 1$. First, we use a well-known inequality that states
	\begin{equation}
	a^y \leq 1 + (a - 1)y,
	\end{equation}
	for $a \geq 1$ and $y \in [0,1]$. Using $a = 2$ and $y = n / (n + x)$,
	we obtain
	\begin{equation}
	2^\frac{n}{n+x} \leq 1 + \frac{n}{n+x}.
	\end{equation}
	Since $n \geq 1$ and $x \geq 1$, we have
	\begin{equation}
	1 + \frac{n}{n+x} < 1 + \frac{n}{x},
	\end{equation}
	and hence by transitivity,
	\begin{equation}
	2^\frac{n}{n+x} < 1 + \frac{n}{x}.
	\end{equation}
	Taking logarithms on both sides yields
	\begin{equation}
	\frac{n}{n+x} < \log\left(1 + \frac{n}{x}\right),
	\end{equation}
	which implies the desired result.
\end{proof}

\balance

\end{document}